\documentclass[11p,reqno]{amsart}

\usepackage[T1]{fontenc}
\usepackage{graphicx}
\usepackage{amssymb,amsthm,amsmath,mathrsfs,bm,braket,marginnote}
\usepackage{enumerate}
\usepackage{appendix}
\usepackage[colorlinks=true, pdfstartview=FitV, linkcolor=blue, citecolor=blue, urlcolor=blue]{hyperref}

\usepackage{pgf}
\usepackage{pgfplots}
\usepackage{tikz}
\usetikzlibrary{arrows,calc}
\usepackage{verbatim}
\usetikzlibrary{decorations.pathreplacing,decorations.pathmorphing}

\numberwithin{equation}{section}
\newtheorem{theorem}{Theorem}[section]
\newtheorem{lemma}[theorem]{Lemma}
\newtheorem{coro}[theorem]{Corollary}
\newtheorem{prop}[theorem]{Property}

\theoremstyle{definition}
\newtheorem{define}[theorem]{Definition}
\newtheorem{remark}[theorem]{Remark}

 \reversemarginpar
\begin{document}

\title[Novikov peakons and B-Toda lattices]{An Application of Pfaffians to multipeakons of the Novikov equation and the finite Toda lattice of BKP type}
\author{Xiang-Ke Chang}
\address{LSEC, ICMSEC, Academy of Mathematics and Systems Science, Chinese Academy of Sciences, P.O.Box 2719, Beijing 100190, PR China; and School of Mathematical Sciences, University of Chinese Academy of Sciences, Beijing 100049, PR China}
\email{changxk@lsec.cc.ac.cn}
\thanks{Corresponding author: Shi-Hao Li}


\author{Xing-Biao Hu}
\address{LSEC, ICMSEC, Academy of Mathematics and Systems Science, Chinese Academy of Sciences, P.O.Box 2719, Beijing 100190, PR China; and School of Mathematical Sciences, University of Chinese Academy of Sciences, Beijing 100049, PR China}
\email{hxb@lsec.cc.ac.cn}


\author{Shi-Hao Li}
\address{LSEC, ICMSEC, Academy of Mathematics and Systems Science, Chinese Academy of Sciences, P.O.Box 2719, Beijing 100190, PR China; and School of Mathematical Sciences, University of Chinese Academy of Sciences, Beijing 100049, PR China}
\email{lishihao@lsec.cc.ac.cn}

\author{Jun-Xiao Zhao}
\address{School of Mathematical Sciences, University of Chinese Academy of Sciences, Beijing 100049, PR China; and Glasgow International College, 56 Dumbarton Road, Glasgow, G11 6NU, UK}
\email{jxzhao@ucas.ac.cn; junxiao.zhao@kaplan.com}
\date{}

\dedicatory{}


\begin{abstract}
The Novikov equation is an integrable analogue of the Camassa-Holm equation with a cubic (rather than quadratic) nonlinear term. Both these equations support a special family of weak solutions called multipeakon solutions.  In this paper, an approach involving Pfaffians  is applied to study multipeakons of the Novikov equation. First, we show that the Novikov peakon ODEs describe an isospectral flow on the manifold cut out by certain Pfaffian identities. Then, a link between the Novikov peakons and the finite Toda lattice of BKP type (B-Toda lattice) is established based on the use of Pfaffians.  Finally, certain generalizations of the Novikov equation and the finite B-Toda lattice are proposed together with special solutions. To our knowledge, it is the first time that the peakon problem is interpreted 
in terms of Pfaffians.

\end{abstract}

\maketitle
\tableofcontents
\section{Introduction}

It is essential to search for simple mathematical models for effectively describing nonlinear phenomena in nature, for example, the breakdown of regularity. Whitham, in his book \cite{whitham1974linear},  emphasized the need for a water wave model exhibiting a soliton interaction, the existence of peaked waves, and, at the same time, allowing for breaking waves. In 1993, Camassa and Holm \cite{camassa1993integrable} derived such a shallow water wave model
\begin{equation}\label{eq:ch}
m_t+(um)_x+mu_x=0, \qquad m=u-u_{xx},
\end{equation}
nowadays known as the Camassa-Holm (CH) equation.   The CH equation admits, in particular, peakon solutions (simply called peakons) as its solitary wave solutions with peaks. In the case of $n$ peaks the solution takes the form
\[
u=\sum_{j=1}^nm_j(t)e^{-|x-x_j(t)|}.
\]
Clearly, the peakons are smooth solutions except at the peak positions $x = x_j (t )$ where the $x$ derivative of $u$ is discontinuous, forcing us to interpret peakons  in a suitable weak sense. 
The mathematics of peakons has attracted a great deal of attention since they were discovered in the CH equation.  It would be nearly impossible 
to mention all relevant contributions so we select only those which are 
directly relevant to the present paper.  
\begin{enumerate}
\item The dynamics of the CH peakons can be described by an ODE system, which was explicitly solved by the inverse spectral method in \cite{beals1999multipeakons,beals2000multipeakons} by Beals, Sattinger and Szmigielski. The closed form can be expressed in terms of Hankel determinants and orthogonal polynomials.
\item The CH peakons are orbitally stable, in the sense that a solution that initially is close to a peakon solution is also close to some peakon solution at a later time. This was proved by Constantin and Strauss \cite{constantin2000stability,constantin2002stability}.
\item There exists a bijection between a certain Krein type string problem, 
which in turn is equivalent to the CH peakon problem,  and the Jacobi spectral problem, the isospectral flow of which gives the celebrated Toda lattice \cite{beals2001peakons}. In a well defined sense,  the CH peakon ODEs could be regarded as a negative flow of the finite Toda lattice.
\end{enumerate}

Following the discovery of the CH equation, other similar equations 
were found.  The list includes: 
the Degasperis-Procesi (DP) equation \cite{degasperis1999asymptotic}, the Novikov equation \cite{hone2008integrable,novikov2009generalisations} (discovered by Vladimir Novikov), Geng-Xue (GX) equation \cite{geng2009extension} and other generalizations of CH equation \cite{chen2006two,geng2011three,li2014four,song2011new} etc., all of which fall into a category of Camassa-Holm type equations. This family of equations shares one characteristic feature: all these equations support peakon solutions. 

 Among these equations, the DP and Novikov equations have been studied the most. In particular, the DP and Novikov peakons can also be explicitly constructed by employing the inverse spectral method \cite{hone2009explicit,lundmark2005degasperis} involving 
interesting generalizations of the classical inhomogeneous string, namely a  cubic string and a dual cubic string respectively. Eventually, the solutions are expressed in terms of some bimoment determinants associated with Cauchy biorthogonal polynomials. It is worth noting that the DP and Novikov peakons are orbitally stable \cite{liu2,liu2014stability,liu1}. Therefore, as for DP and Novikov peakons, the first two aspects listed above have been well studied but the question of some connection to the Toda lattice has not been addressed so far. 

Recall that Hankel determinants appear in the expressions of multipeakons of the CH equation \cite{beals1999multipeakons}. If one substitutes a multipeakon solution into the CH peakon ODEs and computes all necessary  derivatives of the corresponding determinants, the CH peakon ODEs are nothing but certain determinant identities; this result was proven in  \cite{chang2014generalized}. In other words,  the CH peakon ODEs can be regarded as an isospectral flow on the manifold cut out by determinant identities. Moreover, the CH equation has a reciprocal link  to the first negative flow in the Korteweg-de Vries (KdV) hierarchy \cite{fuchssteiner1996some} which in turn belongs to a larger class of flows, called the AKP type \cite{hirota2004direct,jimbo1983solitons}, whose closed form solutions in terms of  the $\tau$-function are all determinantal solutions. It therefore  makes sense to expect similar determinantal description for the CH peakons. 

By contrast, the DP equation is connected with a negative flow in the Kaup-Kupershmidt (KK) hierarchy \cite{degasperis2002new}, while the Novikov equation is related to a negative flow in the Sawada-Kotera(SK) hierarchy via reciprocal transformations \cite{hone2008integrable}. The SK hierarchy belongs to BKP type, while KK hierarchy belongs to CKP type. (Note that the letters "A,B,C" refer to different types of infinite dimensional Lie groups and the corresponding Lie algebras which are associated with respective hierarchies \cite{date1981kp,jimbo1983solitons}. The transformation groups of AKP, BKP, CKP type equations are $GL(\infty)$, $O(\infty)$ and $Sp(\infty)$, respectively.)

 Due to much more complicated structures than the AKP type equations, the solutions of the BKP type equations are expressed not as determinants but as Pfaffians \cite{hirota2004direct}  (See Appendix \ref{app_pf} for more details on the Pfaffian.). Therefore,  it is natural to guess that Pfaffians might be 
 appropriate objects to describe multipeakons of the Novikov equation. After trial and error, we succeeded in rewriting its peakons solutions in terms of Pfaffians and thanks to the use of Pfaffians we were able to obtain some results on the third point above (the Toda lattice interpretation). 
  
 In summary, this paper concerns  the Novikov equation and its multipeakons interpreted in terms of Pfaffians. The novel aspects include:
\begin{enumerate}
\item \textbf{The Pfaffian perspective on the peakon problem 
is formulated for the first time.}
\item \textbf{The Novikov peakon ODEs are shown to be connected with the finite Toda lattice of BKP type (B-Toda lattice). More precisely, the Novikov peakon ODEs can be viewed as a negative flow of the finite B-Toda lattice, generalizing earlier results of \cite{beals2001peakons}. }
\item \textbf{Some generalizations of the Novikov equation and the finite B-Toda lattice are proposed together with special solutions.}
\end{enumerate}

Our paper is organized as follows: In Section \ref{sec:NV_pf}, we 
provide basic information on the Novikov equation and its peakon solutions, and then cast the peakon solutions in the language of Pfaffians, from which the Novikov peakon ODEs are interpreted as an isospectral flow on a  manifold restricted by Pfaffian identities. A correspondence between the Novikov peakon ODEs and finite B-Toda lattice together with their solutions is given in Section \ref{sec:NV_btoda}.  In Section \ref{sec:gen}, generalizations of the Novikov equation and finite B-Toda lattices are proposed together with their solutions. Section \ref{sec:con} is devoted to conclusion and discussions.

 
 \section{Novikov peakons}\label{sec:NV_pf}
In this section, we give some background on the Novikov equation and its multipeakons. Then we describe how to use Pfaffians in the peakon problem.

The Novikov equation
 \begin{equation}\label{eq:NV}
m_t+m_xu^2+3muu_x=0,\qquad m=u-u_{xx}
\end{equation}
is an integrable system with cubic nonlinearity, which was derived by V. Novikov \cite{novikov2009generalisations} in a symmetry classification of nonlocal partial differential equations and first published in the paper by Hone and Wang \cite{hone2008integrable}. It possesses a matrix Lax pair \cite{hone2008integrable} and admits a bi-Hamiltonian structure. There also exist infinitely many conservation quantities. 

The Lax pair of the Novikov equation \cite{hone2008integrable} reads:

 \begin{equation}
D_x \Phi=U
\Phi, \qquad D_t \Phi=V\Phi\label{NV_Lax}
\end{equation}
with 
\begin{equation*}
U=\begin{pmatrix}
0&\lambda m &1\\
0&0&\lambda m\\
1&0&0
\end{pmatrix},\quad
V=\begin{pmatrix}
-uu_x&{\lambda }^{-1}{u_x}-\lambda u^2m& u_x^2\\
{\lambda }^{-1}{u}&-\lambda^{-2} &-{\lambda }^{-1}{u_x}-\lambda u^2m\\
-u^2&{\lambda }^{-1}{u}& uu_x
\end{pmatrix}.
\end{equation*}
 In other words, the compatibility condition 
$$(D_xD_t- D_tD_x )\Phi= 0$$ 
implies the zero curvature condition
$$ U_t-V_x+[U,V]=0,$$
which is exactly the Novikov equation.

The Novikov equation  \eqref{eq:NV} 
admits the multipeakon solution of the form
\begin{equation*}
u=\sum_{k=1}^n m_k(t)e^{-|x-x_k(t)|}
\end{equation*}
in some weak sense if the positions and momenta satisfy the following ODE system:
\begin{align}
\dot x_{k}&=u(x_k)^2, \qquad \dot m_{k}=-m_ku(x_k)\langle u_x\rangle (x_k), \qquad 1\leq k \leq n,\label{NV_eq:peakon} 
\end{align}
where $\langle f \rangle(a)$ denotes the average of left and right limits at the point $a$.

 For the positive momenta $m_k(0)$, the ODE system \eqref{NV_eq:peakon} was solved by Hone, Lundmark and Szmigielski in \cite{hone2009explicit} who used the inverse spectral method to explicitly construct the peakon solutions.

  \subsection{Review of the work by Hone et. al.} 
Let's sketch in this subsection the basic ideas used to construct  Novikov peakons.
  
  For brevity, we will follow the notations in \cite{hone2009explicit} throughout the paper. Given an index set $I$ and a positive integer $K$, the following abbreviations are used:
  \begin{equation*}
    b_I=\prod_{j\in I}b_j,
    \qquad
    \zeta_I^l=\prod_{j\in I}\zeta_j^l,
    \qquad
    \Delta_I=\prod_{i,j\in I, \, i<j}(\zeta_j-\zeta_i),\qquad
    \Gamma_I=\prod_{i,j\in I, \, i<j}(\zeta_j+\zeta_i).
  \end{equation*}
Besides, $\binom{[1,K]}{k}$ denotes the set of $k$-element subsets
  $J=\{j_1<\dots<j_k\}$ of the integer interval
  $[1,K]=\{1,\dots,K\}$, i.e.,
  \begin{equation*}
    \binom{[1,K]}{k}=\{J=\{j_1,\dots,j_k\} : 1\leq j_1<\dots<j_k\leq K\}.
  \end{equation*}

In the peakon sector, $m$ can be seen as a discrete measure
$$
mdx=2\sum_{k=1}^nm_k\delta(x-x_k)dx
$$
with positive $m_k$.
Due to the integrability of the Novikov equation, it was shown that the peakon ODE system \eqref{NV_eq:peakon} can be thought to be an isospectral evolution system (see Appendix B in \cite{hone2009explicit} for details) so that the peakon ODE system \eqref{NV_eq:peakon} can be solved explicitly by use of inverse spectral method. 

More precisely, they considered the spectral problem on a finite interval $[-1,1]$ (a dual cubic string) obtained from the original problem by  a Liouville transformation. In the peakon sector, $m$ is transformed into 
$$g(y)dy=2\sum_{k=1}^ng_k\delta(y-y_k)dy,$$
where
\begin{equation}
g_k=m_k\cosh x_k, \ \ \ y_k=\tanh x_k.
\end{equation}
By introducing two Weyl functions, which encode the spectral data $\{\zeta_k,b_k\}_{k=1}^n$, it was proved that  the string data $\{g_k,y_k\}_{k=1}^n$ 
restricted by
$$g_k>0,\qquad -1=y_0<y_1<y_2<\cdots<y_n<y_{n+1}=1$$
and $\{\zeta_k,b_k\}_{k=1}^n$ 
satisfying
$$b_k>0,\qquad 0<\zeta_1<\zeta_2<\cdots<\zeta_n$$
is in a bijective correspondence.
The inverse mapping is given by 
\begin{equation*}
y_{k'}=\frac{Z_k-W_{k-1}}{Z_k+W_{k-1}},\qquad g_{k'}=\frac{1}{2}\frac{Z_k+W_{k-1}}{U_kU_{k-1}}
\end{equation*}
leading to
\begin{equation*}
l_{k'-1}=y_{k'}-y_{k'-1}=\frac{2(U_k)^4}{(Z_k+W_{k-1})(Z_{k+1}+W_{k})},
\end{equation*}
where
\begin{align*}
& W_k=V_kU_k-V_{k-1}U_{k+1},\qquad \qquad Z_k=T_kU_k-T_{k+1}U_{k-1},\\
&T_k=\sum_{I\in{\left(\substack{[1,n]\\k}\right)}}\frac{\Delta_I^2}{\zeta_I\Gamma_I}b_I, \quad U_k=\sum_{I\in{\left(\substack{[1,n]\\k}\right)}}\frac{\Delta_I^2}{\Gamma_I}b_I, \quad  V_k=\sum_{I\in{\left(\substack{[1,n]\\k}\right)}}\frac{\Delta_I^2}{\Gamma_I}\zeta_Ib_I.
\end{align*}
Note that the index $k'=n+1-k$ is used for simplicity, and $T_0=U_0=V_0=1$, $T_k=U_k=V_k=0$ for $k>n$.  

According to the $t$ part of the Lax pair,   $\{\zeta_k,b_k\}_{k=1}^n$ evolve as
$$\dot\zeta_k=0,\qquad \dot b_k=\frac{b_k}{\zeta_k},$$
so that $\zeta_k$ remain positive and $b_k(t)=b_k(0)e^{\frac{t}{\zeta_k}}>0$.

In summary, the conclusion can be stated as follows:
\begin{theorem}[Hone, Lundmark \& Szmigielski \cite{hone2009explicit}]\label{th:NV_solution}
The Novikov equation admits the n-peakon solution of the form
\begin{equation*}
u=\sum_{k=1}^n m_k(t)e^{-|x-x_k(t)|},
\end{equation*}
where
\begin{equation}\label{sol:NVxt_form}
x_{k'}=\frac{1}{2}\log\frac{Z_k}{W_{k-1}},\qquad m_{k'}=\frac{\sqrt{Z_kW_{k-1}}}{U_kU_{k-1}},
\end{equation}
for the index $k'=n+1-k,\ k=1,2,\ldots, n$. Here
\begin{align*}
& W_k=V_kU_k-V_{k-1}U_{k+1},\qquad \qquad Z_k=T_kU_k-T_{k+1}U_{k-1},\\
&T_k=\sum_{I\in{\left(\substack{[1,n]\\k}\right)}}\frac{\Delta_I^2}{\zeta_I\Gamma_I}b_I, \quad U_k=\sum_{I\in{\left(\substack{[1,n]\\k}\right)}}\frac{\Delta_I^2}{\Gamma_I}b_I, \quad  V_k=\sum_{I\in{\left(\substack{[1,n]\\k}\right)}}\frac{\Delta_I^2}{\Gamma_I}\zeta_Ib_I
\end{align*}
with the constants $\zeta_j$ and $b_j(t)$ satisfying
\begin{equation}\label{NV_evl:b}
0<\zeta_1<\zeta_2<\cdots\zeta_n,\qquad 
b_j(t)=b_j(0)e^{\frac{t}{\zeta_j}}>0.
\end{equation}
\end{theorem}
\begin{remark}
Let 
$$
\mathcal{P}=\{x_1<x_2<\cdots<x_n,\ \ \  m_j>0,\ \ \  j=1,2,\cdots,n\}.
$$
 It was shown that, if the initial data are in $\mathcal{P}$, then $x_j(t),m_j(t)$ will exist for all the time $t\in \bf R$ under the peakon flow \eqref{NV_eq:peakon}, and stay in $\mathcal{P}$. In other words, the solution in Theorem \ref{th:NV_solution}  is global.
\end{remark}
\begin{remark}
For the positive momenta  $m_j$, one will get the pure peakon solution. One can also formulate the result for the negative momenta by
simply changing the sign of $m_j$, leading to what is called  the pure antipeakon case.
\end{remark}

 \subsection{Expressions in terms of Pfaffians}
 It is shown in \cite{hone2008integrable} that the Novikov equation \eqref{eq:NV} is related by a reciprocal transformation to a negative flow in the Sawada-Kotera hierarchy, which belongs to BKP type equations. Usually, the solutions of BKP type equations are expressed as Pfaffians (see the introduction for Pfaffians in Appendix \ref{app_pf}), so that the Pfaffian identities may play a role in the verification of the solutions. Thus
 it is natural to expect that the peakon solutions can be rewritten in terms of Pfaffians. Thanks to the de Bruijn's formula \eqref{id_deBr1}-\eqref{id_deBr2} and the Schur's identities \eqref{id_schur}-\eqref{id_schur_odd}, we succeeded in expressing $T_k,U_k,V_k$ in the formulae of the Novikov peakon solution in terms of Pfaffians. 
 
 
 First of all, we observe the following result, which involves Pfaffians. 
 
 \begin{lemma}\label{lem:debr}
 Let 
 $$f_k(t)=\sum_{j=1}^n(\zeta_j)^kb_j(t),\qquad\qquad  0<\zeta_1<\cdots<\zeta_n, \quad b_j(t)=b_j(0)e^{\frac{t}{\zeta_j}}>0.$$ 
 Introduce
the Pfaffian entries 
$$
(d_0,i)=f_i(t),\qquad (i,j)
={\iint\limits_{-\infty<t_1<t_2<t}}\left[ f_{i-1}(t_2)f_{j-1}(t_1)-f_{i-1}(t_1)f_{j-1}(t_2)\right] dt_1 dt_2,$$ 
so that Pfaffians can be defined according to Definition \ref{def:pfa}. \footnote{Note that there are many different notations for the Pfaffians (See Definition \ref{def:pfa}, Remarks \ref{rem:pf1}, \ref{rem:pf2}). It is convenient to present the Pfaffian identities in the notation used in this paper. This notation will be mostly used throughout the whole paper except a small number of places, where other notations will be employed for convenience. Hopefully no confusion will arise from these changes of notations.}

Then we have, 
\begin{enumerate}
\item  for  $k\in \mathbb{N}$, $l\in \mathbb Z$,
\begin{align*}
&(-1)^k(l+2,l+3,\cdots, l+2k+1)=\idotsint \limits_{-\infty<t_1<\cdots<t_{2k}<t} \det [f_{l+i}(t_j)]_{i,j=1,\cdots,2k}dt_1\cdots dt_{2k},\\
&(-1)^k(d_0,l+2,l+3,\cdots, l+2k)=\idotsint\limits_{-\infty<t_1<\cdots<t_{2k-1}<t} \det [f_{l+i}(t_j)]_{i,j=1,\cdots,2k-1}dt_1\cdots dt_{2k-1},
\end{align*}
\item for  $1\leq k\leq n$, $l\in \mathbb Z$,
\begin{align*}
&\idotsint \limits_{-\infty<t_1<\cdots<t_{k}<t} \det [f_{l+i}(t_j)]_{i,j=1,\cdots,k}dt_1\cdots dt_{k}=(-1)^\frac{k(k-1)}{2}\sum_{J\in{\left(\substack{[1,n]\\k}\right)}}\zeta_J^{l+2}b_J\frac{\Delta_J^2}{\Gamma_J},
\end{align*}
\item  for $k>n,$ $l\in \mathbb Z$,
$$\det (f_{l+i}(t_j))_{i,j=1,\cdots,k}=0.$$
\end{enumerate}
 \end{lemma}
 \begin{proof}
 Without loss of generality, it suffices to give the proof for $l=0$. 
 
The first conclusion immediately follows by employing the de Bruijn' formulae \eqref{id_deBr1}-\eqref{id_deBr2}.
 
Observe that
 \[
  (f_{i}(t_j))_{i,j=1,\cdots,k}=
 \begin{pmatrix}
\zeta_1&\zeta_2&\cdots&\zeta_{n}\\
\zeta_1^2&\zeta_2^2&\cdots&\zeta_n^2\\
\vdots&\vdots&\ddots&\vdots\\
\zeta_1^{k}&\zeta_2^{k}&\cdots&\zeta_n^{k}
\end{pmatrix}\cdot
\begin{pmatrix}
b_1(t_1)&b_1(t_2)&\cdots&b_1(t_k)\\
b_2(t_1)&b_2(t_2)&\cdots&b_2(t_k)\\
\vdots&\vdots&\ddots&\vdots\\
b_{n}(t_1)&b_{n}(t_2)&\cdots&b_{n}(t_k)
\end{pmatrix},
 \]
 which implies that the rank of the matrix $(f_{i}(t_j))_{i,j=1,\cdots,k}$ is not greater than $n$. Thus the third conclusion is confirmed.
  
 Finally, let's prove that the second conclusion is valid. We shall give a detailed proof for the even case. The proof of the odd case is similar and we omit it.
 
We notice that
\begin{align}
&\idotsint \limits_{-\infty<t_1<\cdots<t_{2k}<t} \det [f_{i}(t_j)]_{i,j=1,\cdots,2k}dt_1\cdots dt_{2k} \nonumber\\
=\ &\idotsint \limits_{-\infty<t_1<\cdots<t_{2k}<t} \det\left(
\begin{pmatrix}
\zeta_1&\zeta_2&\cdots&\zeta_{n}\\
\zeta_1^2&\zeta_2^2&\cdots&\zeta_n^2\\
\vdots&\vdots&\ddots&\vdots\\
\zeta_1^{2k}&\zeta_2^{2k}&\cdots&\zeta_n^{2k}
\end{pmatrix}
\begin{pmatrix}
b_1(t_1)&b_1(t_2)&\cdots&b_1(t_{2k})\\
b_2(t_1)&b_2(t_2)&\cdots&b_2(t_{2k})\\
\vdots&\vdots&\ddots&\vdots\\
b_{n}(t_1)&b_{n}(t_2)&\cdots&b_{n}(t_{2k})
\end{pmatrix}
\right)dt_1\cdots dt_{2k}\nonumber\\
=\ &\sum_{J\in\left(\substack{[1,n]\\2k}\right)}\zeta_J\cdot\Delta_J\cdot\mathcal B_J, \label{BJ1}
\end{align}
where
\begin{align*}
\mathcal B_J=
\idotsint \limits_{-\infty<t_1<\cdots<t_{2k}<t}\left|
\begin{array}{cccc}
b_{j_1}(t_1)&b_{j_1}(t_2)&\cdots&b_{j_1}(t_{2k})\\
b_{j_2}(t_1)&b_{j_2}(t_2)&\cdots&b_{j_2}(t_{2k})\\
\vdots&\vdots&\ddots&\vdots\\
b_{j_{2k}}(t_1)&b_{j_{2k}}(t_2)&\cdots&b_{j_{2k}}(t_{2k})
\end{array}
\right|dt_1\cdots dt_{2k},
\end{align*}
for some fixed $J=\{j_1,\dots,j_{2k}\}, 1\leq j_1<\dots<j_{2k}\leq n$.
Here the last step follows from the Cauchy-Binet formula and the property of Vandermonde matrix.

Applying the de Bruijn's formula \eqref{id_deBr1} for the even case to $\mathcal B_J$, we obtain
\begin{equation*}
\mathcal B_J=(j_1,\cdots,j_{2k}), \quad \text{where}\quad (j_m,j_l)=\iint\limits_{-\infty<t_1<t_2<t}b_{j_m}(t_1)b_{j_l}(t_2)-b_{j_l}(t_1)b_{j_m}(t_2)dt_1dt_2.\footnote{Here we refer  the reader to Definition \ref{def:pfa}, Remarks \ref{rem:pf1}, \ref{rem:pf2} in order to understand the meaning of this Pfaffian.
}
\end{equation*}
In fact, this Pfaffian entry can be explicitly computed by using
$$
b_i(t)=b_i(0)e^{\frac{t}{\zeta_i}},
$$
leading to
$$(j_m,j_l)=\frac{\zeta_{j_m}-\zeta_{j_l}}{\zeta_{j_m}+\zeta_{j_l}}\zeta_{j_m}\zeta_{j_l}b_{j_m}b_{j_l}.$$
 
Consequently,
\[
\mathcal B_J=\zeta_Jb_J\cdot (j_1^*,\cdots,j_{2k}^*), \quad \text{where}\quad 
(j_m^*,j_l^*)=\frac{\zeta_{j_m}-\zeta_{j_l}}{\zeta_{j_m}+\zeta_{j_l}}.\footnote{Again, see Definition \ref{def:pfa}, Remarks \ref{rem:pf1}, \ref{rem:pf2}.}
\]
By use of the Schur's Pfaffian identity \eqref{id_schur} for the even case, we get
\begin{equation}
\mathcal B_J=(-1)^{k(2k-1)}\zeta_Jb_J\frac{\Delta_J}{\Gamma_J}. \label{BJ_pf}
\end{equation}
Combining \eqref{BJ1} and \eqref{BJ_pf}, we eventually arrive at the desired formula. Thus the proof is finished.

\end{proof}

Comparing the expressions of $T_k,U_k,V_k$ in Theorem \ref{th:NV_solution} with the above lemma and noting the conventions, it is obvious that the following corollary follows. 
\begin{coro} \label{coro:TUV_pf}
For $k\in \mathbb Z$, the following relations hold,
\begin{align*}
&T_{2k}=(-1,0,\cdots,2k-2),  &&T_{2k-1}=(d_0,-1,0,\cdots,2k-3),\\
&U_{2k}=(0,1,\cdots,2k-1), &&U_{2k-1}=(d_0,0,1,\cdots,2k-2),\\
&V_{2k}=(1,2,\cdots,2k),&&V_{2k-1}=(d_0,1,2,\cdots,2k-1),
\end{align*}
where the Pfaffian entries are those in Lemma \ref{lem:debr}. We use the convention that the Pfaffian of order $0$ is set to be $1$ and negative order to be 0.


\end{coro}
So far, we have reformulated  the multipeakon solution of the Novikov equations in the language of Pfaffians. Actually, we can employ the formulae on Pfaffians in Appendix \ref{app_pf} to prove the validity of Theorem \ref{th:NV_solution}, which is the object of next subsection.

\subsection{An alternative proof of Theorem \ref{th:NV_solution}} 
To begin with, we reveal some underlying properties among  $T_k,U_k,V_k,W_k,Z_k$. 

\begin{lemma}\label{lem:sum_wz}
Under the definitions in Lemma \ref{lem:debr}, for $0\leq k\leq n$, the following nonlinear identities hold,
\begin{align}
\sum_{j=1}^{k}\frac{W_{j-1}}{U_jU_{j-1}}=\frac{V_{k-1}}{U_{k}},\qquad \sum_{j=k+1}^{n}\frac{Z_{j}}{U_jU_{j-1}}=\frac{T_{k+1}}{U_{k}}\label{nonl_id}.
\end{align}
\end{lemma}
\begin{proof}
The conclusion is obvious for $k=0$ in the first formula.
Substituting the expression 
$W_j=V_jU_j-V_{j-1}U_{j+1}$, we have 
\[
\sum_{j=1}^{k}\frac{W_{j-1}}{U_jU_{j-1}}=\sum_{j=1}^{k}\frac{V_{j-1}U_{j-1}-V_{j-2}U_{j}}{U_jU_{j-1}}=\sum_{j=1}^{k}\left(\frac{V_{j-1}}{U_{j}}-\frac{V_{j-2}}{U_{j-1}}\right)=
\frac{V_{k-1}}{U_{k}},
\]
where $V_{-1}=0,U_0=1$ are used in the last step.

Similarly, employing $Z_k=T_kU_k-T_{k+1}U_{k-1}$, the second identity can be proved  by noticing $T_{n+1}=0$.
\end{proof}

Besides, we have some derivative formulae.
\begin{lemma}\label{lemma_der_pf}
Under the definitions in Lemma \ref{lem:debr}, for $k\in \mathbb N, l\in \mathbb Z$, the following relations hold,
\begin{align*}
\frac{d}{dt}(l+2,l+3,\cdots,l+2k+1)&=(l+1,l+3,l+4,\cdots,l+2k+1)\\
&=(d_0,d_{-1},l+2,l+3,\cdots,l+2k+1),\\
\frac{d}{dt}(d_0,l+2,l+3,\cdots,l+2k)&=(d_0,l+1,l+3,l+4,\cdots,l+2k)\\
&=(d_{-1},l+2,l+3,\cdots,l+2k),
\end{align*}
where $(d_{-1},i)=(d_0,i-1)=f_{i-1},\ (d_0,d_{-1})=0$.
 Note that there exist two different forms for every derivative. We will refer to the former one as Wronski-type derivative and the latter one as Gram-type derivative (see Remark \eqref{rem:wg}).
 \end{lemma}
\begin{proof}
Recall that 
\[
f_i(t)=\sum_{j=1}^n(\zeta_j)^ib_j(t), \qquad \text{with}\qquad \dot b_j=\frac{b_j}{\zeta_j},
\]
which immediately gives 
\[
\dot f_i=f_{i-1}.
\]
In other words, we have
\begin{equation*}
\frac{d}{dt}(d_0,i)=(d_{-1},i).
\end{equation*}
Furthermore, it is not hard to prove that 
\begin{equation*}
\frac{d}{dt}(i,j)=(i-1,j)+(i,j-1)=(d_0,d_{-1},i,j).
\end{equation*}
Once the above derivatives for the Pfaffian entries hold, the result can be easily confirmed with the help of the derivative formulae \eqref{der1}-\eqref{der2_2}. 
\end{proof}

The above lemma results in the following corollary. 
\begin{coro}\label{coro_NV_der_uvt}
For $0\leq k\leq n$, the following relations hold,
\begin{align*}
\dot U_{k+1}U_{k}-U_{k+1}\dot U_{k}&=T_{k+1}V_{k}-V_{k-1}T_{k+2},\\
\dot U_k V_k-  U_k \dot V_k&=\dot U_{k+1}V_{k-1}-U_{k+1}\dot V_{k-1},\\
\dot T_kU_k- T_k\dot U_k&=\dot T_{k+1}U_{k-1}-T_{k+1}\dot U_{k-1},
\end{align*}
and
\[
\dot W_k=2(\dot U_k V_k-\dot U_{k+1} V_{k-1}),\qquad 
\dot Z_k=2(T_k\dot U_k-T_{k+1}\dot U_{k-1}).\]
\end{coro}
\begin{proof}
We will proceed the proof for $k=2m$, and the odd case is similar.

By using the Wronski-type derivative in Lemma \ref{lemma_der_pf},
the first formula is nothing but a Pfaffian identity in \eqref{pf1}-\eqref{pf2}.

As for the second formula, we will see it is no other than a Pfaffian identity in \eqref{pf1}-\eqref{pf2} if we employ the Gram-type derivative in Lemma \ref{lemma_der_pf}.

The third identity immediately follows since its structure is an analogue of the second one. The only difference arises from the index shift. 

The last two are just matters of combining the expressions of $W_k, Z_k$, and the second and third formulae, respectively.
\end{proof}

Now we are ready to confirm the validity of Theorem \ref{th:NV_solution} by using the Pfaffian technique. 
\begin{proof}[A proof of Th. \ref{th:NV_solution} in terms of Pfaffians]
It is clear that what we need to prove is
\begin{align}
\dot x_{k'}&=u(x_{k'})^2\nonumber\\
&=\left(\sum_{j=1}^{k-1}m_{j'}e^{x_{k'}-x_{j'}}+\sum_{j=k}^nm_{j'}e^{x_{j'}-x_{k'}}\right)^2 \label{NV_x_ode},\\
\dot m_{k'}&=-m_{k'}u(x_{k'})\langle u_x(x_{k'})\rangle\nonumber\\
&=m_{k'}\left(\sum_{j=1}^{k-1}m_{j'}e^{x_{k'}-x_{j'}}+\sum_{j=k}^nm_{j'}e^{x_{j'}-x_{k'}}\right)\left(\sum_{j=k+1}^nm_{j'}e^{x_{j'}-x_{k'}}-\sum_{j=1}^{k-1}m_{j'}e^{x_{k'}-x_{j'}}\right).\label{NV_m_ode}
\end{align}
Again note the notation $k'=n+1-k$.

Substituting the expressions of $x_k$ and $m_k$ in terms of $W_k, V_k, Z_k$, after simplification, we arrive at
\begin{align}
&\frac{1}{2}\left(\dot Z_kW_{k-1}-\dot W_{k-1}Z_k\right)=\left(\sum_{j=1}^{k-1}\frac{W_{j-1}}{U_jU_{j-1}}Z_k+\sum_{j=k}^{n}\frac{Z_{j}}{U_jU_{j-1}}W_{k-1}\right)^2, \label{NV_zw_ode1}\\
&\frac{1}{2}\frac{\dot Z_k W_{k-1}+Z_k\dot W_{k-1}}{Z_kW_{k-1}}-\frac{\dot U_k U_{k-1}+U_k\dot U_{k-1}}{U_kU_{k-1}}\nonumber\\
&=\frac{1}{Z_kW_{k-1}}\left(\sum_{j=1}^{k-1}\frac{W_{j-1}}{U_jU_{j-1}}Z_k+\sum_{j=k}^{n}\frac{Z_{j}}{U_jU_{j-1}}W_{k-1}\right)\left(\sum_{j=k+1}^{n}\frac{Z_{j}}{U_jU_{j-1}}W_{k-1}-\sum_{j=1}^{k-1}\frac{W_{j-1}}{U_jU_{j-1}}Z_k\right). \label{NV_zw_ode2}
\end{align}

Let's consider \eqref{NV_zw_ode1} first.  By using the expressions of $Z_k, W_k$ in terms of $U_k, T_k,V_k$ and Corollary \ref{coro_NV_der_uvt},  we see that the left hand of the equality reads:
\begin{align*}
&\frac{1}{2}\left(\dot Z_kW_{k-1}-\dot W_{k-1}Z_k\right)\\
=&(T_k\dot U_k-T_{k+1}\dot U_{k-1})W_{k-1}-Z_k(\dot U_{k-1}V_{k-1}-\dot U_k V_{k-2})\\
=&\dot U_k(W_{k-1}T_k+Z_kV_{k-2})-\dot U_{k-1}(W_{k-1}T_{k+1}+Z_kV_{k-1})\\
=&(\dot U_k U_{k-1}-U_k \dot U_{k-1})(T_k V_{k-1}-T_{k+1}V_{k-2})\\
=&(T_k V_{k-1}-T_{k+1}V_{k-2})^2,
\end{align*}
which equals to the right hand by using the expansions of $W_k$ and $Z_k$, and Lemma \ref{lem:sum_wz}. Thus \eqref{NV_zw_ode1} is valid.

Now we turn to the proof of \eqref{NV_zw_ode2}.
\begin{align*}
L&=\frac{Z_k\dot W_{k-1}+(T_kV_{k-1}-V_{k-2}T_{k+1})^2}{Z_kW_{k-1}}-\frac{2U_k\dot U_{k-1}-(V_{k-2}T_{k+1}-V_{k-1}T_k)}{U_kU_{k-1}},\\
R&=\frac{1}{Z_kW_{k-1}}\left(\frac{V_{k-2}Z_k+T_kW_{k-1}}{U_{k-1}}\right)\left(\frac{T_{k+1}W_{k-1}}{U_k}-\frac{V_{k-2}Z_{k}}{U_{k-1}}\right)\\
&=\frac{1}{Z_kW_{k-1}}(T_kV_{k-1}-T_{k+1}V_{k-2})\left[(T_kV_{k-1}-T_{k+1}V_{k-2})-Z_k\left(\frac{V_{k-1}}{U_k}+\frac{V_{k-2}}{U_{k-1}}\right)\right]\\
&=\frac{1}{Z_kW_{k-1}}(T_kV_{k-1}-T_{k+1}V_{k-2})\left[(T_kV_{k-1}-T_{k+1}V_{k-2})-Z_k\left(\frac{W_{k-1}+2U_{k}V_{k-2}}{U_kU_{k-1}}\right)\right],
\end{align*}
which gives
\begin{align*}
&L-R\\
=\ &\frac{\dot W_{k-1}}{W_{k-1}}-2\frac{\dot U_{k-1}}{U_{k-1}}+2\frac{(T_kV_{k-1}-T_{k+1}V_{k-2})V_{k-2}}{W_{k-1}U_{k-1}}\\
=\ &\frac{2}{W_{k-1}U_{k-1}}\left[(\dot U_{k-1}V_{k-1}-\dot U_k V_{k-2})U_{k-1}-(U_{k-1}V_{k-1}-U_kV_{k-2})\dot U_{k-1}-V_{k-2}(V_{k-2}T_{k+1}-V_{k-1}T_k)\right]\\
=\ &\frac{2V_{k-2}}{W_{k-1}U_{k-1}}\left[\dot U_{k-1} U_k-\dot U_k U_{k-1}-V_{k-2}T_{k+1}+V_{k-1}T_k\right].
\end{align*}
The above expression vanishes once the first identity in Corollary \ref{coro_NV_der_uvt} is employed. Therefore the proof is completed.
\end{proof}

So far, the Pfaffian has been successfully applied to reformulating the multipeakon solutions of the Novikov equation. As a consequence, the proof above implies that the Novikov peakon ODEs  can be viewed as an isospectral flow on the manifold cut out by Pfaffian identities.

\section{Novikov peakons and finite B-Toda lattices}\label{sec:NV_btoda}

Recall that, the CH equation is related to the first negative flow in the hierarchy of the Korteweg-de Vries(KdV) equation, and the Novikov equation to a negative flow in the Sawada-Kotera(SK) hierarchy \cite{hone2008integrable}, via some reciprocal transformations. The KdV equation belongs to AKP type and the SK equation is of BKP type.

On the other hand, it is shown in \cite{beals2001peakons} that, there is a close connection between the CH peakon ODEs and the finite Toda lattice, which is an AKP type equation. 
 Thus we suspect there are some relations between the Novikov peakon ODEs \eqref{NV_eq:peakon} and some BKP type lattice, which is the subject of this section.

As is indicated in Section \ref{sec:NV_pf}, by use of inverse spectral method, the Novikov peakon ODE system \eqref{NV_eq:peakon} is linearised to
\[
\dot \zeta_j(t)=0 \qquad \dot b_j(t)=\frac{b_j(t)}{\zeta_j}.
\]
Consequently, 
\[
f_i(t)=\sum_{j=1}^n(\zeta_j)^ib_j(t)
\]
undergoes the evolution
\begin{equation*}
\dot f_i=f_{i-1}.
\end{equation*}
In some sense, this can be viewed as a negative flow.

We plan to start from a linearised flow in positive direction by introducing
\begin{equation}
f_i(t)=\sum_{j=1}^n(\zeta_j)^ib_j(t) \label{exp:g}
\end{equation}
with 
\begin{align*}
&\dot \zeta_j(t)=0, && 0<\zeta_1<\cdots<\zeta_n,\\
&\dot b_j(t)=\zeta_jb_j(t), &&b_j(0)>0,
\end{align*}
yielding
\begin{equation}
\dot f_i=f_{i+1}.\label{evo:g}
\end{equation}
Then we would like to seek a nonlinear ODE system after "dressing clothes". 
Eventually, we find that the finite B-Toda equation \cite{gilson2003two,hirota2001soliton,miwa1982hirota} as follows can be obtained:
\begin{equation}\label{eq:btoda}
\dot p_k=p_k(q_{k+1}-2q_{k}+q_{k-1}), \qquad \dot q_k=(q_{k+1}-q_{k-1})p_k,\qquad k=1,2,\ldots,n,
\end{equation}
with the boundary $p_0=p_n=q_0=p_nq_{n+1}=0.$
\begin{remark}
Here the nonlinear form \eqref{eq:btoda} is recognized as B-Toda lattice since it is similar to the well-known Toda lattice \cite{beals2001peakons}. Besides, as we'll see soon, it has a bilinear form, which is nothing but the bilinear B-Toda lattice studied in \cite{gilson2003two,hirota2001soliton,miwa1982hirota}.
\end{remark}
\subsection{Finite B-Toda lattice} In this subsection, we will present how to derive the finite B-Toda lattice after "dressing clothes" step by step.

First of all, let's introduce the tau functions $\{\tau_k\}_{k \in \mathbb Z}$.
\begin{define} \label{def:tauk}
For any integer $k$, define
\begin{equation}
\tau_k=\sum_{J\in{\left(\substack{[1,n]\\k}\right)}}b_J\frac{\Delta_J^2}{\Gamma_J},
\end{equation}
with
  \begin{equation*}
   b_J=\prod_{j\in J}b_j,
    \qquad
    \zeta_J^l=\prod_{j\in J}\zeta_j^l,
    \qquad
    \Delta_J=\prod_{i,j\in J, \, i<j}(\zeta_j-\zeta_i),\qquad
   \Gamma_J=\prod_{i,j\in J, \, i<j}(\zeta_j+\zeta_i),
  \end{equation*}
  for some index set $J$, and 
  \begin{align*}
&\dot \zeta_j(t)=0, && 0<\zeta_1<\cdots<\zeta_n,\\
&\dot b_j(t)=\zeta_jb_j(t), &&b_j(0)>0.
\end{align*} Obviously, we have
  \[
 \tau_0=1,\qquad \tau_k=0, \qquad k<0 \ \  \text{or}\ \  k>n.
  \]

  \end{define}
    
We will see that $\{\tau_k\}_{k\in \mathbb Z}$ satisfy an ODE system, as is indicated in Theorem \ref{th:ode:tau}, before which, it is necessary to give the following statements. 
\begin{lemma}\label{lem:deb_btoda}
If the Pfaffian entries are defined as 
\begin{equation}
( d_0, i)=f_i(t),\qquad ( i, j)={\iint\limits_{-\infty<t_1<t_2<t}}\left[ f_{i+1}(t_1)f_{j+1}(t_2)-f_{i+1}(t_2)f_{j+1}(t_1)\right] dt_1 dt_2,\label{exp:tau_pf}
\end{equation}
with \eqref{exp:g},
 then
 \begin{enumerate}
\item  for  $k\in \mathbb{N}$, $l\in \mathbb Z$,
\begin{align*}
&( l,{l+1},\cdots, {l+2k-1})=\idotsint \limits_{-\infty<t_1<\cdots<t_{2k}<t} \det [f_{l+i}(t_j)]_{i,j=1,\cdots,2k}dt_1\cdots dt_{2k},\\
&(d_0,l,l+1,\cdots, l+2k-2)=\idotsint\limits_{-\infty<t_1<\cdots<t_{2k-1}<t} \det [f_{l+i}(t_j)]_{i,j=1,\cdots,2k-1}dt_1\cdots dt_{2k-1},
\end{align*}
\item for  $1\leq k\leq n$, $l\in \mathbb Z$,
\begin{align*}
&\idotsint \limits_{-\infty<t_1<\cdots<t_{k}<t} \det [f_{l+i}(t_j)]_{i,j=1,\cdots,k}dt_1\cdots dt_{k}=\sum_{J\in{\left(\substack{[1,n]\\k}\right)}}\zeta_J^lb_J\frac{\Delta_J^2}{\Gamma_J},
\end{align*}
\item  for $k>n,$ $l\in \mathbb Z$,
$$\det (f_{l+i}(t_j))_{i,j=1,\cdots,k}=0.$$
\end{enumerate}
\end{lemma}
\begin{proof}
The argument is almost the same as Lemma \ref{lem:debr}. The first conclusion is nothing but a consequence of applying  de Bruijn' formulae \eqref{id_deBr1}-\eqref{id_deBr2}. The third conclusion immediately follows by noticing the matrix decomposition.

Regarding the second result, we sketch its proof for the even case. The argument for the odd case is similar. 

Without loss of generality, take $l=0$.  By following the computation in the proof of Lemma \ref{lem:debr}, it is not hard to obtain

\begin{align}
&\idotsint \limits_{-\infty<t_1<\cdots<t_{2k}<t} \det [f_{i}(t_j)]_{i,j=1,\cdots,2k}dt_1\cdots dt_{2k}=\sum_{J\in\left(\substack{[1,n]\\2k}\right)}\zeta_J\cdot\Delta_J\cdot\mathcal C_J,\label{CJ1}
\end{align}
where
$$\mathcal C_J=(j_1,\cdots,j_{2k}),\quad \text{with}\quad 
(j_m,j_l)=\iint\limits_{-\infty<t_1<t_2<t}b_{j_m}(t_1)b_{j_l}(t_2)-b_{j_l}(t_1)b_{j_m}(t_2)dt_1dt_2.$$
By using
$$
b_i(t)=b_i(0)e^{\zeta_it},
$$
we can explicitly work out 
$$(j_m,j_l)=\frac{\zeta_{j_l}-\zeta_{j_m}}{\zeta_{j_m}+\zeta_{j_l}}\frac{b_{j_m}b_{j_l}}{\zeta_{j_m}\zeta_{j_l}},$$
which leads to
\[\mathcal C_J=\frac{b_J}{\zeta_J}\cdot (j_1^*,\cdots,j_{2k}^*),\qquad \text{with}\qquad  
(j_m^*,j_l^*)=\frac{\zeta_{j_l}-\zeta_{j_m}}{\zeta_{j_m}+\zeta_{j_l}}.
\]
By use of the Schur's Pfaffian identity \eqref{id_schur}, we get
\begin{equation}
\mathcal C_J=\frac{b_J}{\zeta_J}\frac{\Delta_J}{\Gamma_J}. \label{CJ_pf}
\end{equation}
Combining \eqref{CJ1} and \eqref{CJ_pf}, we eventually confirm the identity.
\end{proof}

From the above lemma, it is clear that
\begin{coro} \label{coro:tau_pfa}
For any $k\in \mathbb Z$, 
\begin{equation*}
\tau_{2k}=(0,\cdots,2k-1),\qquad\tau_{2k-1}=(d_0,0,\cdots,2k-2),
\end{equation*}
where
\begin{equation*}
(d_0,i)=f_i(t),\qquad ( i, j)={\iint\limits_{-\infty<t_1<t_2<t}}\left[ f_{i+1}(t_1)f_{j+1}(t_2)-f_{i+1}(t_2)f_{j+1}(t_1)\right] dt_1 dt_2,
\end{equation*}
with
 \[
f_i(t)=\sum_{j=1}^n(\zeta_j)^ib_j(t),\qquad \qquad 
 0<\zeta_1<\cdots<\zeta_n, \qquad b_j(t)=b_j(0)e^{\zeta_jt}>0.\]
\end{coro}

We can also prove that the $\{\tau_k\}_{k \in \mathbb Z}$ undergo the following evolutions.
\begin{lemma} \label{lem:der_btoda}
For any $k\in \mathbb Z$, the following relations hold,
\begin{align}
\dot \tau_{2k}&=(0,\cdots,2k-2,2k)=(d_1,d_0,0,\cdots,2k-1),\label{der:tau1}\\
\dot\tau_{2k-1}&=(d_0,0,\cdots,2k-3,2k-1)=(d_1,0,\cdots,2k-2),\label{der:tau1_odd}\\
\ddot \tau_{2k}&=(d_1,d_0,0,\cdots,2k-2,2k),\label{der:tau2_even}\\
\ddot \tau_{2k-1}&=(d_1,0,\cdots,2k-3,2k-1),\label{der:tau2}
\end{align}
where $(d_1,i)=f_{i+1},\ (d_1,d_0)=0$, and, $(d_0,i)$ and $( i, j)$ are those in  Corollary \ref{coro:tau_pfa}.
\end{lemma}
\begin{proof}
It is easy to see 
\[
\dot f_i=f_{i+1},
\]
which can be equivalently written as
\[
\frac{d}{dt}(d_0,i) =(d_1,i).
\]
Besides, it is not hard to obtain that
\[
\frac{d}{dt}(i,j)=(d_1,d_0,i,j)=(i+1,j)+(i,j+1).
\]
By employing the derivative formulae in \eqref{der1}-\eqref{der2_2}, we obtain \eqref{der:tau1} and \eqref{der:tau1_odd}.

If we apply Wronski-type for the first derivative and Gram-type derivative for the second, then we get \eqref{der:tau2_even} and \eqref{der:tau2} respectively.
\end{proof}

Now we are ready to state the ODE system that $\{\tau_k\}_{k \in \mathbb Z}$ satisfy.
\begin{theorem} \label{th:ode:tau}
$\{\tau_k\}_{k=0}^n$ satisfy the following ODEs:
\begin{equation}\label{bi:tau}
\ddot \tau_k \tau_k-(\dot\tau_k)^2=\tau_{k-1}\dot\tau_{k+1}-\dot\tau_{k-1}\tau_{k+1},\ \ \ k=0,1,\ldots,n,
\end{equation}
with $\tau_{-1}=\tau_{n+1}=0.$
\end{theorem}
\begin{proof}
The boundary values $\tau_{-1}, \tau_{n+1}$ immediately follow by Lemma \ref{lem:deb_btoda} and the convention.

Inserting \eqref{der:tau1}-\eqref{der:tau2} into \eqref{bi:tau} in the even case and odd case respectively, we get
\begin{align*}
&(d_1,d_0,0,\cdots,2k-2,2k)(0,\cdots,2k-1)-(d_1,d_0,0,\cdots,2k-1)(0,\cdots,2k-2,2k)\\
=\ &(d_0,0,\cdots,2k-2)(d_1,0,\cdots,2k)-(d_1,0,\cdots,2k-2)(d_0,0,\cdots,2k),\\
&(d_1,0,\cdots,2k-1,2k+1)(d_0,0,\cdots,2k)-(d_0,0,\cdots,2k-1,2k+1)(d_1,0,\cdots,2k)\\
=\ &(0,\cdots,2k-1)(d_1,d_0,0,\cdots,2k+1)-(d_1,d_0,0,\cdots,2k-1)(0,\cdots,2k+1),
\end{align*}
which are nothing but the Pfaffian identities in \eqref{pf1}-\eqref{pf2}. Therefore the conclusion follows.
\end{proof}

The equation \eqref{bi:tau} is a bilinear form of the finite B-Toda lattice. The bilinear form \eqref{bi:tau} has appeared as a special case of famous discrete BKP equation (also called Hirota-Miwa equation) \cite{gilson2003two,hirota2001soliton,miwa1982hirota}. In fact, from our derivation, we have given an explicit construction for the so-called "molecule solution" to the  finite bilinear B-Toda lattice \eqref{bi:tau}. What's more, the molecule solution of the finite nonlinear B-Toda lattice \eqref{eq:btoda} can also be obtained.

\begin{theorem}\label{th:btoda_sol}
The finite B-Toda equation \eqref{eq:btoda} admits the solution 
\begin{equation}
p_k=\frac{\tau_{k+1}\tau_{k-1}}{(\tau_k)^2}, \qquad  q_k=\frac{d}{dt}\log \tau_k. \label{exp:pq}
\end{equation}
Here
\begin{equation*}
\tau_{2k}=(0,\cdots,2k-1),\qquad\tau_{2k-1}=(d_0,0,\cdots,2k-2),
\end{equation*}
where
\begin{equation*}
(d_0,i)=f_i(t),\qquad ( i, j)={\iint\limits_{-\infty<t_1<t_2<t}}\left[ f_{i+1}(t_1)f_{j+1}(t_2)-f_{i+1}(t_2)f_{j+1}(t_1)\right] dt_1 dt_2
\end{equation*}
with
 \[
f_i(t)=\sum_{j=1}^n(\zeta_j)^ib_j(t), \qquad 0<\zeta_1<\cdots<\zeta_n, \qquad b_j(t)=b_j(0)e^{\zeta_jt}>0.
 \]
\end{theorem}
\begin{proof}
Recall that, Lemma \ref{lem:deb_btoda} tells us
$$\tau_0=1,\tau_{-1}=\tau_{n+1}=0,$$
and 
$$\tau_k>0,\qquad k=1,2,\ldots,n,$$
which is compatible with the boundary condition of the finite B-toda lattice \eqref{eq:btoda}.

It is obvious that the first equation in \eqref{eq:btoda} holds if we substitute the expression \eqref{exp:pq} into it. In the following, we will confirm the validity of the second equation, which is equivalent to verifying
\[
\ddot \tau_k \tau_k-(\dot\tau_k)^2=\tau_{k-1}\dot\tau_{k+1}-\dot\tau_{k-1}\tau_{k+1}.
\]
Clearly, this is the conclusion of Theorem \ref{th:ode:tau}. Therefore, the proof is completed.
\end{proof}
At the end of this subsection, we would like to introduce the variables 
\begin{equation}
\sigma_k=\sum_{J\in{\left(\substack{[1,n]\\k}\right)}}\zeta_JB_J\frac{\Delta_J^2}{\Gamma_J},\qquad k\in \mathbb{Z},
\end{equation}
which admit similar properties to $\{\tau_k\}_{k \in \mathbb Z}$. It is not hard to get
\begin{theorem}\label{th:dotsigma}
For any $k\in \mathbb Z$,
\begin{equation*}
\sigma_{2k}=(1,\cdots,2k),\qquad\sigma_{2k-1}=(d_0,1,\cdots,2k-1),
\end{equation*}
with $\sigma_0=1,\sigma_k=0$ for $k<0 \ \text{or}\ \ k>n,$ and
\begin{align*}
\dot \sigma_{2k}&=(1,\cdots,2k-1,2k+1)=(d_1,d_0,1,\cdots,2k),\\
\dot\sigma_{2k-1}&=(d_0,1,\cdots,2k-2,2k)=(d_1,1,\cdots,2k-1),\\
\ddot \sigma_{2k}&=(d_1,d_0,1,\cdots,2k-1,2k+1),\\
\ddot \sigma_{2k-1}&=(d_1,1,\cdots,2k-2,2k),\\
\ddot \sigma_k \sigma_k&-(\dot\sigma_k)^2=\sigma_{k-1}\dot\sigma_{k+1}-\dot\sigma_{k-1}\sigma_{k+1}.
\end{align*}
where the Pfaffian entries are those in Lemma \ref{lem:der_btoda}.
\end{theorem}
\begin{proof}
The proof can be achieved by following the steps for $\{\tau_k\}_{k \in \mathbb Z}$ in a parallel way.
\end{proof}
As we will see, $\{\sigma_k\}_{k \in \mathbb Z}$ will be involved in the next subsection.

\subsection{Finite modified B-Toda lattice} In order to interpret  how the positive and negative flows are involved, it is helpful to introduce a modification of the B-Toda lattice \eqref{eq:btoda}, which we call a modified B-Toda lattice. The finite modified B-Toda lattice reads
\begin{align} \label{eq:mbtoda}
&\dot  r_k=\frac{r_k(\dot s_{k-1}+\dot s_k)}{s_{k-1}-s_{k}},\qquad \dot s_k=\frac{r_{k+1}}{r_k}(s_{k-1}-s_{k+1}), \qquad k=1,2,\ldots,n,
\end{align}
with $s_0=r_{n+1}=0$ and $r_{n+1}s_{n+1}\neq 0.$ 

We also have its solution as is described below.
\begin{theorem}\label{th:mbtoda}
The finite modified B-Toda equation \eqref{eq:mbtoda} admits the solution
\begin{align}
r_k=\frac{\tau_{k}}{\tau_{k-1}},\qquad s_k=\frac{\sigma_{k-1}}{\tau_k}. \label{exp:rs}
\end{align}

\end{theorem}
\begin{proof}
As for the first equation, it is equivalent to prove
\[
s_k\left(\frac{\dot r_k}{r_k}+\frac{\dot s_k}{s_k}\right)=s_{k-1}\left(\frac{\dot r_k}{r_k}-\frac{\dot s_{k-1}}{s_{k-1}}\right).
\]
The above equation can be rewritten as
\[
s_k\frac{d}{dt}{\log(r_ks_k)}=s_{k-1}\frac{d}{dt}(\log\frac{r_{k}}{s_{k-1}}),
\]
which leads to
\[
\dot\sigma_{k-1}\tau_{k-1}-\sigma_{k-1}\dot\tau_{k-1}=\sigma_{k-2}\dot\tau_{k}-\dot\sigma_{k-2}\tau_{k}.
\]
If the Pfaffian expression of each term is plug into these equations, they become Pfaffian identities in \eqref{pf1}-\eqref{pf2} in the odd and even cases, respectively. Thus, they are valid.

Now let's turn to the second equation. 
If we substitute the expressions \eqref{exp:rs} into  the second equation in \eqref{eq:mbtoda}, it is easy to see that it suffices to prove 
\[
\dot \sigma_{k-1}\tau_k-\sigma_{k-1}\dot\tau_k=\sigma_{k-2}\tau_{k+1}-\sigma_{k}\tau_{k-1}.
\]
Consider the odd and even cases respectively. These are nothing but Pfaffian identities in \eqref{pf1}-\eqref{pf2} if $\{\tau_k\}_{k \in \mathbb Z}$ and $\{\sigma_k\}_{k \in \mathbb Z}$ and their derivatives  are expressed in terms of Pfaffians.

\end{proof}

\subsection{Positive and negative flows} Now we are ready to give an interpretation on the relation between Novikov peakon ODEs and the finite modified B-Toda lattice, which is summarized as the following theorem\footnote{In fact, there exists a straightforward correspondence between the Novikov peakons and finite B-Toda lattices, which is presented in \cite[Appendix B]{chang2018degasperis}.}.

\begin{theorem}\label{th:nv_btoda}
Let 
\begin{align*}
& \quad  W_k(0)=V_k(0)U_k(0)-V_{k-1}(0)U_{k+1}(0),\qquad \qquad Z_k(0)=T_k(0)U_k(0)-T_{k+1}(0)U_{k-1}(0),\\
&T_k(0)=\left(\sum_{I\in{\left(\substack{[1,n]\\k}\right)}}\frac{\Delta_I^2}{\zeta_I\Gamma_I}b_I\right)_{t=0}, \ U_k(0)=\left(\sum_{I\in{\left(\substack{[1,n]\\k}\right)}}\frac{\Delta_I^2}{\Gamma_I}b_I\right)_{t=0}, \  V_k(0)=\left(\sum_{I\in{\left(\substack{[1,n]\\k}\right)}}\frac{\Delta_I^2}{\Gamma_I}\zeta_Ib_I\right)_{t=0},
\end{align*}
 with
 \[
 0<\zeta_1(0)<\zeta_2(0)<\cdots<\zeta_n(0),\qquad b_j(0)>0.
 \]
 \begin{enumerate}
 \item Introduce the variables $\{x_k(0),m_k(0)\}_{k=1}^n$ defined by
$$ x_{k'}(0)=\frac{1}{2}\log\frac{Z_k(0)}{W_{k-1}(0)},\qquad m_{k'}(0)=\frac{\sqrt{Z_k(0)W_{k-1}(0)}}{U_k(0)U_{k-1}(0),}$$
where $k'=n+1-k.$ 

If $\{\zeta_j(t),b_j(t)\}_{j=1}^n$ evolve as 
 \[
 \dot \zeta_j=0,\qquad \dot b_j=\frac{b_j}{\zeta},
 \]
then $\{x_k(t),m_k(t)\}_{k=1}^n$  satisfy the Novikov peakon ODEs \eqref{NV_eq:peakon}.
\item 
 Introduce the variables $\{r_k(0),s_k(0)\}_{k=1}^n$ defined by
 \[
r_k(0)=\frac{U_{k}(0)}{U_{k-1}(0)},\qquad s_k(0)=\frac{V_{k-1}(0)}{U_k(0)}. 
\]
 If $\{\zeta_j(t),b_j(t)\}_{j=1}^n$ evolve as 
 \[
 \dot \zeta_j=0,\qquad \dot b_j=\zeta_jb_j,
 \]
 then $\{r_k(t),s_k(t)\}_{k=1}^n$  satisfy the finite modified B-Toda lattice \eqref{eq:mbtoda}.
 \item There exists a mapping from $\{x_k(0),m_k(0)\}_{k=1}^n$ to $\{r_k(0),s_k(0)\}_{k=1}^n$ according to
\begin{align*}
 &r_k(0)=\frac{e^{x_n(0)}}{m_n(0)}\prod_{j=1}^{k-1}\frac{1}{2m_{j'}(0)m_{j'-1}(0)\cosh x_{j'}(0)\cosh x_{j'-1}(0)(\tanh x_{j'}(0)-\tanh x_{j'-1}(0))}, \\
 &s_k(0)=\sum_{j=1}^{k}m_{j'(0)}e^{-x_{j'}(0)}.
\end{align*}
 \end{enumerate}
 \end{theorem}
  \begin{proof}
 The first two immediately follow as they are the subjects of Theorem \ref{th:NV_solution} and Theorem \ref{th:mbtoda}.
 
The rest part is devoted to the proof for the third result.
Introduce the variables 
\[
y_{j'}(0)=\frac{Z_j(0)-W_{j-1}(0)}{Z_j(0)+W_{j-1}(0)},\qquad g_{j'}(0)=\frac{1}{2}\frac{Z_j(0)+W_{j-1}(0)}{U_j(0)U_{j-1}(0)},
\]
and 
\[
l_{j'}(0)=y_{j'}(0)-y_{j'-1}(0)=\frac{2(U_j(0))^4}{(Z_j(0)+W_{j-1}(0))(Z_{j+1}(0)+W_{j}(0))},
\]
based on Theorem \ref{th:NV_solution}.
It is not hard to see that
\[
\frac{1}{2g_{j'}(0)g_{j'-1}(0)(y_{j'}(0)-y_{j'-1}(0))}=\frac{U_{j+1}(0)U_{j-1}(0)}{U_{j}(0)^2},
\]
which gives
\[
\prod_{j=1}^{k-1}\frac{1}{2g_{j'}(0)g_{j'-1}(0)(y_{j'}(0)-y_{j'-1}(0))}=\frac{U_{k}(0)}{U_{k-1}(0)U_1(0)}.
\]
Recall that 
\[
g_j(0)=m_j(0)\cosh x_j(0), \ \ \ y_j(0)=\tanh x_j(0),
\]
and 
\[
m_{n}(0)e^{-x_n(0)}=\frac{1}{U_1(0)}.
\]
Then the conclusion on $r_k(0)$  in terms of $x_k(0)$ and $m_k(0)$ follows.

Furthermore, 
\[
\sum_{j=1}^{k}m_{j'}(0)e^{-x_{j'}(0)}=\sum_{j=1}^{k}\frac{W_{j-1}(0)}{U_{j}(0)U_{j-1}(0)}=\frac{V_{k-1}(0)}{U_k(0)}=s_k(0),
\]
where the second identity follows from Lemma \ref{lem:sum_wz}.

Therefore we complete the proof.
 \end{proof}
 
Theorem \ref{th:nv_btoda} implies that it is reasonable to regard  the Novikov peakon system and the finite modified B-Toda lattices as opposite flows.
 
 \section{Generalizations}\label{sec:gen}
 In this section, we shall give generalizations of the finite (modified) B-Toda lattices and Novikov peakon lattice, consequently Novikov equation together with their solutions.  An arbitrary nonnegative parameter $\alpha$ will be introduced. As special cases, all the corresponding results in the previous sections are covered when $\alpha=0$. 
 
 Our motivation is that there has existed a generalized bilinear B-Toda lattice in \cite{gilson2003two,hirota2001soliton}. Our goal was to  investigate its molecule solution and we eventually succeeded.  Based on the relation between the finite B-Toda lattice and the Novikov peakon lattice, we can also give a generalization of the Novikov peakon lattice, and, subsequently, a generalized Novikov equation. In the following, we shall present our results by proving theorems regarding the generalized equations together with their solutions rather than the derivation of these 
 facts. Some of the proofs are omitted because they can be filled up trivially as analogs to those in the previous sections. 
 \subsection{Generalized finite B-Toda lattices}
 
 \subsubsection{Generalized finite bilinear B-Toda lattice}
 We begin with a family of  $\{\tilde\tau_k\}_{k\in \mathbb Z}$ expressed in terms of Pfaffians, whose Pfaffian entries involve a parameter $\alpha$.  As we will see, $\{\tilde\tau_k\}_{k\in \mathbb Z}$ satisfy the generalized finite bilinear B-Toda lattice \eqref{gbi:tau}.
 \begin{define} \label{def:tilde_tau}
 For any $k\in \mathbb Z$, define
 \begin{equation*}
\tilde \tau_{2k}=(0,\cdots,2k-1),\qquad\tilde \tau_{2k-1}=(d_0,0,\cdots,2k-2),
\end{equation*}
where
\begin{equation*}
(d_0,i)=f_i(t),\qquad ( i, j)=e^{-\alpha t}{\iint\limits_{-\infty<t_1<t_2<t}}e^{\alpha t_2}\left[ f_{i+1}(t_1)f_{j+1}(t_2)-f_{i+1}(t_2)f_{j+1}(t_1)\right] dt_1 dt_2,
\end{equation*}
with
 \[
f_i(t)=\sum_{j=1}^n(\zeta_j)^ib_j(t),
\qquad\qquad 0<\zeta_1<\cdots<\zeta_n, \qquad b_j(t)=b_j(0)e^{\zeta_jt}>0.\]
Note that we use the convention $\tilde \tau_0=1,\tilde\tau_k=0$ for $k<0$.
 \end{define}

The following lemma  tells us that $\tilde\tau_k$ for $1\leq k\leq n$ are positive and $\tilde\tau_{n+1}=0.$ To this end, we will use the minor-summation formula \eqref{id:minor_even}-\eqref{id:minor_odd} due to Ishikawa and Wakayama \cite{ishikawa1995minor}.
\begin{lemma}\label{lem:gdet_btoda}
$\{\tilde\tau_k\}_{k\in \mathbb Z}$ admit the expressions
\begin{align*}
\tilde\tau_k=\left\{
\begin{array}{ll}
\sum\limits_{J\in{\left(\substack{[1,n]\\k}\right)}}b_J\frac{\Delta_J^2}{\tilde\Gamma_J},&1\leq k\leq n,\\
0,&k>n,
\end{array}
\right.
\end{align*}
with
  \begin{equation*}
   b_J=\prod_{j\in J}b_j,
    \qquad
    \Delta_J=\prod_{i,j\in J, \, i<j}(\zeta_j-\zeta_i),\qquad
  \tilde \Gamma_J=\prod_{i,j\in J, \, i<j}(\zeta_j+\zeta_i+\alpha).
  \end{equation*}
\end{lemma}
\begin{proof}
We shall first present the detailed proof for the even case. 

Let $Q=(q_{i,j})_{1\leq i,j\leq 2k}$ be a skew-symmetric matrix with the entries
\[
q_{i,j}= (i-1, j-1)=e^{-\alpha t}{\iint\limits_{-\infty<t_1<t_2<t}}e^{\alpha t_2}\left[ f_{i}(t_1)f_{j}(t_2)-f_{i}(t_2)f_{j}(t_1)\right] dt_1 dt_2,
\]
which implies 
\[
Pf(Q)=\tilde\tau_{2k}.
\]
It is not hard to confirm that 
\[
Q=RAR^T,
\]
where 
\begin{align*}
R=
\begin{pmatrix}
\zeta_1&\zeta_2&\cdots&\zeta_{n}\\
\zeta_1^2&\zeta_2^2&\cdots&\zeta_n^2\\
\vdots&\vdots&\ddots&\vdots\\
\zeta_1^{2k}&\zeta_2^{2k}&\cdots&\zeta_n^{2k}
\end{pmatrix}_{2k\times n},
\qquad
A=
\begin{pmatrix}
a_{11}&a_{12}&\cdots&a_{1n}\\
a_{21}&a_{22}&\cdots&a_{2n}\\
\vdots&\vdots&\ddots&\vdots\\
a_{n1}&a_{n2}&\cdots&a_{nn}\\\end{pmatrix}_{n\times n},
\end{align*}
with
\begin{align*}
a_{ij}&=e^{-\alpha t}{\iint\limits_{-\infty<t_1<t_2<t}}e^{\alpha t_2}\left[ b_{i}(t_1)b_{j}(t_2)-b_{i}(t_2)b_{j}(t_1)\right] dt_1 dt_2.
\end{align*}
In fact,  A is a skew-symmetric matrix and its entries can be computed explicitly
\[
a_{ij}=\frac{\zeta_j-\zeta_i}{\zeta_i+\zeta_j+\alpha}\frac{b_i(t)b_j(t)}{\zeta_i\zeta_j}.
\]

If $2k>n$, we immediately get
\[
\tilde\tau_{2k}=Pf(Q)=0
\]
by examining the rank of $Q$.

As for $2k\leq n$, by employing the minor-summation formula \eqref{id:minor_even}, it is not hard to get
\[
\tilde\tau_{2k}=Pf(Q)=\sum_{J\in\left(\substack{[1,n]\\2k}\right)}Pf(A_J^J)\det(R_J).
\]
Here
\[
Pf(A_J^J)=\frac{b_J\Delta_J}{\zeta_J\tilde\Gamma_J},\qquad \det(R_J)=\zeta_J\Delta_J,
\]
the first one of which is obtained by using a similar argument for the identity \eqref{CJ_pf}  with the help of the Schur's Pfaffian identity \eqref{id_schur}.
Thus the proof for the even case is completed.

For the odd case, it is noted that 
\[
(\tilde\tau_{2k-1})^2=\tilde R \tilde A\tilde R^T,
\]
where
\begin{align*}
\tilde R=
\begin{pmatrix}
1&0&0&\cdots&0\\
0&\zeta_1&\zeta_2&\cdots&\zeta_{n}\\
0&\zeta_1^2&\zeta_2^2&\cdots&\zeta_n^2\\
\vdots&\vdots&\vdots&\ddots&\vdots\\
0&\zeta_1^{2k-1}&\zeta_2^{2k-1}&\cdots&\zeta_n^{2k-1}
\end{pmatrix}_{(2k)\times (n+1)},
\qquad
\tilde A=
\begin{pmatrix}
0&\frac{b_1}{\zeta_1}&\frac{b_2}{\zeta_2}&\cdots&\frac{b_n}{\zeta_n}\\
-\frac{b_1}{\zeta_1}&a_{11}&a_{12}&\cdots&a_{1n}\\
-\frac{b_2}{\zeta_2}&a_{21}&a_{22}&\cdots&a_{2n}\\
\vdots&\vdots&\vdots&\ddots&\vdots\\
-\frac{b_n}{\zeta_n}&a_{n1}&a_{n2}&\cdots&a_{nn}\\\end{pmatrix}_{ (n+1)\times (n+1)},
\end{align*}
with
\begin{align*}
a_{ij}&=e^{-\alpha t}{\iint\limits_{-\infty<t_1<t_2<t}}e^{\alpha t_2}\left[ b_{i}(t_1)b_{j}(t_2)-b_{i}(t_2)b_{j}(t_1)\right] dt_1 dt_2.
\end{align*}
Then the proof can be trivially achieved by following the even case and using the minor-summation formula \eqref{id:minor_odd}.
\end{proof}

\begin{remark}
 It is noted that although de Bruijn's formula is employed in the argument for the special case in Section \ref{sec:NV_btoda}, the formula doesn't apply here. This proof also implies an alternative proof for the corresponding conclusion for $\{\tau_k\}_{k \in \mathbb Z}$ in Section \ref{sec:NV_btoda}.
\end{remark}

The following lemma is analogous to Lemma \ref{lem:der_btoda}. 
\begin{lemma} \label{lem:gder_btoda}
For any $k\in \mathbb Z$, the following relations hold,
\begin{align*}
\dot {\tilde\tau}_{2k}&=(0,\cdots,2k-2,2k)=(d_1,d_0,0,\cdots,2k-1)-\alpha k(0,1,\cdots,2k-1),\\
\dot{\tilde\tau}_{2k-1}&=(d_0,0,\cdots,2k-3,2k-1)=(d_1,0,\cdots,2k-2)-\alpha (k-1)(d_0,0,\cdots,2k-2),\\
\ddot {\tilde\tau}_{2k}&=(d_1,d_0,0,\cdots,2k-2,2k)-\alpha k(0,\cdots,2k-2,2k),\\
\ddot {\tilde\tau}_{2k-1}&=(d_1,0,\cdots,2k-3,2k-1)-\alpha (k-1)(d_0,0,\cdots,2k-3,2k-1),
\end{align*}
where $(d_1,i)=f_{i+1},\  (d_1,d_0)=0$, and, $(d_0,i)$ and $(i,j)$ are those in Definition \ref{def:tilde_tau}.
\end{lemma}
\begin{proof}
Note that 
\[
\frac{d}{dt}(i,j)=(d_1,d_0,i,j)-\alpha (i,j)=(i+1,j)+(i,j+1).
\]
The result follows by applying similar argument in the proof of Lemma \ref{lem:der_btoda}.
\end{proof}

Now we are ready to introduce the generalized finite bilinear B-Toda lattice and its solution.
\begin{theorem}\label{gth:b_btoda}
$\{\tilde\tau_k\}_{k=0}^n$ satisfy the generalized finite bilinear B-Toda lattice
\begin{equation}\label{gbi:tau}
\ddot {\tilde\tau}_k \tilde\tau_k-(\dot{\tilde\tau}_k)^2=\dot{\tilde\tau}_{k+1}\tilde\tau_{k-1}-\tilde\tau_{k+1}\dot{\tilde\tau}_{k-1}+\alpha{\tilde\tau}_{k+1}\tilde\tau_{k-1} ,\ \ \ k=0,1,\ldots,n,
\end{equation}
with $\tilde\tau_{-1}=0,\tilde\tau_{n+1}=0.$
\end{theorem}
\begin{proof}
Following the way in the proof of Theorem \ref{th:ode:tau}, we can prove this theorem by use of the conclusions in Lemma \ref{lem:gdet_btoda} and  Lemma \ref{lem:gder_btoda}.\end{proof}

\subsubsection{Generalized finite nonlinear B-Toda lattice}
Our generalized finite nonlinear B-Toda lattice reads
\begin{equation}\label{eq:gnbtoda}
\dot{\tilde p}_k=\tilde p_k(\tilde q_{k+1}-2\tilde q_{k}+\tilde q_{k-1}), \qquad  \dot{\tilde q}_k=(\tilde q_{k+1}-\tilde q_{k-1})\tilde p_k+\alpha \tilde p_k,\qquad k=1,2,\ldots,n,
\end{equation}
with the boundary $\tilde p_0=\tilde p_n=\tilde q_0=\tilde p_n\tilde q_{n+1}=0$. We get the following theorem.
 \begin{theorem}
The generalized finite nonlinear B-Toda lattice \eqref{eq:gnbtoda} admits the following solution:
\begin{equation*}
\tilde p_k=\frac{\tilde\tau_{k+1}\tilde\tau_{k-1}}{(\tilde\tau_k)^2}, \qquad  \tilde q_k=\frac{d}{dt}\log \tilde\tau_k. \label{gexp:pq}
\end{equation*}
\end{theorem}
\begin{proof}
The proof can be achieved by following the proof of Theorem \ref{th:btoda_sol} and using Lemma \ref{lem:gder_btoda}, Theorem \ref{gth:b_btoda}.
\end{proof}

\subsubsection{Generalized finite modified B-Toda lattice}

We also derive a generalized finite modified B-Toda lattice:
\begin{align} \label{eq:gmbtoda}
&\dot{\tilde  r}_k=\frac{\tilde r_k(\dot{\tilde s}_{k-1}+\dot {\tilde s}_k-\alpha \tilde s_{k-1})}{\tilde s_{k-1}-\tilde s_{k}},\qquad \dot{\tilde s}_k=\frac{\tilde r_{k+1}}{\tilde r_k}(\tilde s_{k-1}-\tilde s_{k+1}), \qquad k=1,2,\ldots,n,
\end{align}
with $\tilde s_0=\tilde r_{n+1}=0$ and $\tilde r_{n+1}\tilde s_{n+1}\neq 0$. 
Its solution is described below.
\begin{theorem}\label{th:gmbtoda}
The generalized finite modified B-Toda equation \eqref{eq:gmbtoda} admits the solution
\begin{align}
\tilde r_k=\frac{\tilde \tau_{k}}{\tilde \tau_{k-1}},\qquad \tilde s_k=\frac{\tilde\sigma_{k-1}}{\tilde\tau_k},\label{exp:grs}
\end{align}
where $\{\tilde\sigma_k\}_{k \in \mathbb Z}$ are defined as 
\[
\tilde\sigma_{2k}=(1,\cdots,2k),\qquad \tilde\sigma_{2k-1}=(d_0,1,\cdots,2k-1),
\]
with the same Pfaffian entries as those for $\tilde\tau_k.$
\end{theorem}
\begin{proof}
The proof can be completed by a similar argument to that for Theorem \ref{th:mbtoda}.

There are several items which need attention. Notice that $\{\tilde\sigma_k\}_{k \in \mathbb Z}$ have the same structure as $\{\tilde\tau_k\}_{k \in \mathbb Z}$. In fact,
\[
\tilde\sigma_{n+1}=0
\]
holds, and
\begin{align*}
\dot{\tilde\sigma}_{2k}&=(1,\cdots,2k-1,2k+1)=(d_0,d_1,1,\cdots,2k)-\alpha k(1,\cdots,2k ),\\
\dot{\tilde\sigma}_{2k-1}&=(d_0,1,\cdots,2k-2,2k)=(d_1,1,\cdots,2k-1)-\alpha (k-1)(d_0,1,\cdots,2k-1).
\end{align*}
When one substitutes the expressions \eqref{exp:grs} into \eqref{eq:gmbtoda}, the two equations in \eqref{eq:gmbtoda} are respectively equivalent to
\begin{align*}
\dot{\tilde\sigma}_{k-1}\tilde\tau_{k-1}-\tilde\sigma_{k-1}\dot{\tilde\tau}_{k-1}&=\dot{\tilde\sigma}_{k-2}\tau_{k}-\tilde\sigma_{k-2}\dot{\tilde\tau}_{k}+\alpha\tilde\sigma_{k-2}{\tilde\tau}_{k},\\
\dot{\tilde\sigma}_{k-1}\tilde\tau_k-\tilde\sigma_{k-1}\dot{\tilde\tau}_k&=\tilde\sigma_{k-2}\tilde\tau_{k+1}-\tilde\sigma_{k}\tilde\tau_{k-1},
\end{align*}
which become Pfaffian identities by inserting the Pfaffian expression for each term. 
\end{proof}

 \subsection{Generalized Novikov equation}
 \subsubsection{Generalized Novikov peakon ODEs}
\begin{define}\label{def:tilde_tuv}
For any $k\in \mathbb Z$, define the variables
\begin{align*}
&\tilde T_{2k}=(-1,0,\cdots,2k-2),  &&\tilde T_{2k-1}=(d_0,-1,0,\cdots,2k-3),\\
&\tilde U_{2k}=(0,1,\cdots,2k-1), &&\tilde U_{2k-1}=(d_0,0,1,\cdots,2k-2),\\
&\tilde V_{2k}=(1,2,\cdots,2k),&&\tilde V_{2k-1}=(d_0,1,2,\cdots,2k-1),
\end{align*}
where the Pfaffian entries are defined as
\[
(d_0,i)=f_i(t),\qquad (i,j)=e^{-\alpha t}{\iint\limits_{-\infty<t_1<t_2<t}}e^{\alpha t_2}\left[ f_{i-1}(t_2)f_{j-1}(t_1)-f_{i-1}(t_1)f_{j-1}(t_2)\right] dt_1 dt_2,
\]
with
 \[
f_i(t)=\sum_{j=1}^n(\zeta_j)^ib_j(t),
\qquad\qquad 0<\zeta_1<\cdots<\zeta_n, \qquad b_j(t)=b_j(0)e^{\zeta_jt}>0.\]
Again, we use the convention that the Pfaffian for order $0$ is 1, and those for negative order are 0.
\end{define}

It is noted that $\tilde T_k,\tilde U_k,\tilde V_k$ possess the same properties due to the same structures. We will give some properties of the variables 
$$\tilde P_{2k}^l\triangleq(l,l+1,\cdots,l+2k-1),\qquad \tilde P_{2k-1}^l\triangleq(d_0,l,l+1,\cdots,l+2k-2),$$
which give $\tilde T_k,\tilde U_k,\tilde V_k$ when $l=-1,0,1$ respectively.

First we have
\begin{lemma}\label{lem:gdet_nv}
For $l\in \mathbb Z$,  the following relations hold
\begin{align}
\tilde P_k^l=\left\{
\begin{array}{cl}
\sum\limits_{J\in{\left(\substack{[1,n]\\k}\right)}}\zeta_J^lb_J\frac{\Delta_J^2}{\hat\Gamma_J},&1\leq k\leq n,\\
0,&k>n,
\end{array}
\right.
\end{align}
with
  \begin{equation*}
   b_J=\prod_{j\in J}b_j,
    \qquad
    \zeta_J^l=\prod_{j\in J}\zeta_j^l,
    \qquad
    \Delta_J=\prod_{i,j\in J, \, i<j}(\zeta_j-\zeta_i),\qquad
  \hat\Gamma_J=\prod_{i,j\in J, \, i<j}(\zeta_j+\zeta_i+\alpha\zeta_i\zeta_j).
  \end{equation*}
\end{lemma}
\begin{proof}
The result follows from the similar argument in the proof of Lemma \ref{lem:gdet_btoda}.
\end{proof}

Observing the derivative formula 
\[\frac{d}{dt}(i,j)=(d_0,d_{-1},i,j)-\alpha (i,j)=(i-1,j)+(i,j-1),\]
one can derive the following lemma, which is a generalization of Lemma \ref{lemma_der_pf}.
\begin{lemma}\label{lemma_der_gnv}
For any $l,k\in \mathbb Z$, 
\begin{align*}
\dot{\tilde P}_{2k}^l&=(l-1,l+1,l+2,\cdots,l+2k-1)\\
&=(d_0,d_{-1},l,l+1,\cdots,l+2k-1)-\alpha k (l,l+1,l+2,\cdots,l+2k-1),\\
\dot{\tilde P}_{2k+1}^l&=(d_0,l-1,l+1,l+2,\cdots,l+2k)\\
&=(d_{-1},l,l+1,\cdots,l+2k)-\alpha k(d_0,l,l+1,l+2,\cdots,l+2k),
\end{align*}
 where $(d_{-1},i)=f_{i-1},\ (d_0,d_{-1})=0$, and, $(d_0,i)$ and $(i,j)$ are those in Definition \ref{def:tilde_tuv}.
 \end{lemma}
Moreover, we can easily obtain the following generalization regarding to the evolution of $\tilde T_k,\tilde U_k,\tilde V_k$ according to Pfaffian identities.
\begin{lemma} \label{lem:gder:tuv}
For $0\leq k\leq n$, the following relations hold,
\begin{align*}
&\dot {\tilde U}_{k+1} {\tilde U}_{k}- {\tilde U}_{k+1} \dot{\tilde U}_{k}={\tilde T}_{k+1}{\tilde V}_{k}-{\tilde V}_{k-1}{\tilde T}_{k+2},\\
&\dot {\tilde U}_k {\tilde V}_{k}- {\tilde U}_k \dot{\tilde V}_{k}=\dot {\tilde U}_{k+1} {\tilde V}_{k-1}- {\tilde U}_{k+1} \dot{\tilde V}_{k-1}+\alpha {\tilde U}_{k+1} {\tilde V}_{k-1},\\
&\dot {\tilde T}_k {\tilde U}_{k}- {\tilde T}_k \dot{\tilde U}_{k}=\dot {\tilde T}_{k+1} {\tilde U}_{k-1}- {\tilde T}_{k+1} \dot{\tilde U}_{k-1}+\alpha {\tilde V}_{k+1} {\tilde T}_{k-1}.
\end{align*}
\end{lemma}
If we introduce the variables 
\begin{align*}
& \tilde W_k=\tilde V_k\tilde U_k-\tilde V_{k-1}\tilde U_{k+1},\qquad \qquad \tilde Z_k=\tilde T_k\tilde U_k-\tilde T_{k+1}\tilde U_{k-1},
\end{align*}
it is obvious that we have the following result according to Lemma \ref{lem:sum_wz}.
\begin{lemma}
For $0\leq k\leq n$, the following nonlinear identities hold,
\begin{align*}
\sum_{j=1}^{k}\frac{\tilde W_{j-1}}{\tilde U_j\tilde U_{j-1}}=\frac{\tilde V_{k-1}}{\tilde U_{k}},\qquad \sum_{j=k+1}^{n}\frac{\tilde Z_{j}}{\tilde U_j\tilde U_{j-1}}=\frac{\tilde T_{k+1}}{\tilde U_{k}}.
\end{align*}
\end{lemma}
By use of Lemma \ref{lem:gder:tuv}, we immediately have
\begin{coro}
For $0\leq k\leq n$, the following relations hold,
\begin{align*}
\dot {\tilde W}_k=2(&\dot {\tilde U}_k {\tilde V}_k-\dot {\tilde U}_{k+1} {\tilde V}_{k-1})-\alpha  {\tilde U}_{k+1} {\tilde V}_{k-1},\\
\dot {\tilde Z}_k=2(&{\tilde T}_k\dot {\tilde U}_k-{\tilde T}_{k+1}\dot {\tilde U}_{k-1})+\alpha {\tilde T}_{k+1} {\tilde U}_{k-1}.
\end{align*}
\end{coro}

Now we are ready to present the following generalized result regarding the peakon ODEs.
\begin{theorem}\label{th:gnv_ode_sol}
If we let
\[
{\tilde x}_{k'}=\frac{1}{2}\log\frac{{\tilde Z}_k}{{\tilde W}_{k-1}},\qquad {\tilde m}_{k'}=\frac{\sqrt{{\tilde Z}_k{\tilde W}_{k-1}}}{{\tilde U}_k{\tilde U}_{k-1}}, \qquad k=1,2,\ldots,n,
\]
where $k'=n+1-k$,
then $\{{\tilde x}_k,{\tilde m}_k\}_{k=1}^n$ satisfy
\begin{align*}
\dot {\tilde x}_{k'}&=\left(\sum_{j=1}^{k-1}{\tilde m}_{j'}e^{{\tilde x}_{k'}-{\tilde x}_{j'}}+\sum_{j=k}^n{\tilde m}_{j'}e^{{\tilde x}_{j'}-{\tilde x}_{k'}}\right)^2+\frac{\alpha}{2 {\tilde m}_k'}\left(\sum_{j=1}^{k-1}{\tilde m}_{j'}e^{{\tilde x}_{k'}-{\tilde x}_{j'}}+\sum_{j=k}^n{\tilde m}_{j'}e^{{\tilde x}_{j'}-{\tilde x}_{k'}}\right)-\frac{\alpha}{2} ,\\
\dot {\tilde m}_{k'}
&={\tilde m}_{k'}\left(\sum_{j=1}^{k-1}{\tilde m}_{j'}e^{{\tilde x}_{k'}-{\tilde x}_{j'}}+\sum_{j=k}^n{\tilde m}_{j'}e^{{\tilde x}_{j'}-{\tilde x}_{k'}}\right)\left(\sum_{j=k+1}^n{\tilde m}_{j'}e^{{\tilde x}_{j'}-{\tilde x}_{k'}}-\sum_{j=1}^{k-1}{\tilde m}_{j'}e^{{\tilde x}_{k'}-{\tilde x}_{j'}}\right)\\
&\ \ \ +\frac{\alpha}{2}\left(\sum_{j=k+1}^n{\tilde m}_{j'}e^{{\tilde x}_{j'}-{\tilde x}_{k'}}-\sum_{j=1}^{k-1}{\tilde m}_{j'}e^{{\tilde x}_{k'}-{\tilde x}_{j'}}\right).
\end{align*}
\end{theorem}
\begin{proof}
Comparing with the proofs of \eqref{NV_x_ode} and \eqref{NV_m_ode}, one can fill out the proof by using the lemmas in this section.
\end{proof}
\subsubsection{Generalized Novikov equation}

The generalized Novikov equation we propose reads:
\begin{subequations}\label{eq:gNV}
\begin{align}
&{\tilde m}={\tilde u}-{\tilde u}_{xx},\\
&{\tilde m}_t+({\tilde u}^2{\tilde m})_x+{\tilde u}{\tilde u}_x{\tilde m}+\frac{1}{2}\alpha ({\tilde u}{\tilde w}{\tilde m})_x+\frac{1}{2}\alpha {\tilde u}_x{\tilde w}{\tilde m}-\frac{1}{2}\alpha {\tilde m}_x=0, \\
&({\tilde w}{\tilde m})_t+({\tilde w}{\tilde u}^2{\tilde m})_x+\frac{1}{2}\alpha ({\tilde u}{\tilde w}^2{\tilde m})_x-\frac{1}{2}\alpha ({\tilde w}{\tilde m})_x=0.
\end{align}
\end{subequations}
with the arbitrary nonnegative number $\alpha$.
Obviously, it reduces to the Novikov equation \eqref{eq:NV} when $\alpha=0$ and $\tilde w=0$. 

Assume that 
\[
{\tilde u}=\sum_{k=1}^n {\tilde m}_k(t)e^{-|x-{\tilde x}_k(t)|},\qquad \tilde w=\sum_{k=1}^n\frac{1}{{\tilde m}_k(t)}\mathbf 1_{\{{\tilde x}_k(t)\}}^\epsilon(x),
\]
where the indicator function $\mathbf 1_{\{{\tilde x}_k(t)\}}^\epsilon(x)$ implies
\[\mathbf 1_{\{{\tilde x}_k(t)\}}^\epsilon(x)=\left\{\begin{array}{ll}
1,\qquad x\in({\tilde x}_k-\epsilon,{\tilde x}_k+\epsilon),\\
0,\qquad \text{elsewhere},
\end{array}
\right.
\]
and $\epsilon$ is chosen to be a small enough number so that the intervals $({\tilde x}_k-\epsilon,{\tilde x}_k+\epsilon)$ are nonintersecting with each other.

 Obviously, it follows from the first equation in \eqref{eq:gNV} that
$${\tilde m}dx=2\sum_{k=1}^n\tilde m_k\delta(x-{\tilde x}_k)dx,$$

By use of the distribution calculus \cite{beals2000multipeakons,hone2009explicit}, we obtain the following ODE system from the last two equations in \eqref{eq:gNV}:
\begin{align*}
\dot {\tilde x}_{k}={\tilde u}({\tilde x}_{k})^2+\frac{\alpha}{2}\frac{{\tilde u}({\tilde x}_k)}{{\tilde m}_k}-\frac{1}{2}\alpha,\qquad \dot {\tilde m}_{k}&=-\tilde m_{k}{\tilde u}({\tilde x}_{k})\langle {\tilde u}_x({\tilde x}_{k})\rangle-\frac{1}{2}\alpha\langle {\tilde u}_x({\tilde x}_{k})\rangle,
\end{align*} 
which can be written as the equations in Theorem \ref{th:gnv_ode_sol} when ${\tilde x}_1<{\tilde x}_2<\cdots<{\tilde x}_n$. This implies that we have obtained one (at least local) distribution solution  for the generalized Novikov equation \eqref{eq:gNV}.

We end this section by summarizing the above statement. 
 \begin{theorem}
 The generalized Novikov equation \eqref{eq:gNV} admits the (at least local) distribution solution
\begin{equation*}
{\tilde u}=\sum_{k=1}^n {\tilde m}_k(t)e^{-|x-{\tilde x}_k(t)|},\qquad \tilde w=\sum_{k=1}^n\frac{1}{{\tilde m}_k(t)}\mathbf 1_{\{{\tilde x}_k(t)\}}^\epsilon(x),
\end{equation*}
where
\begin{equation*}
{\tilde x}_{k'}=\frac{1}{2}\log\frac{{\tilde Z}_k}{{\tilde W}_{k-1}},\qquad {\tilde m}_{k'}=\frac{\sqrt{{\tilde Z}_k{\tilde W}_{k-1}}}{{\tilde U}_k{\tilde U}_{k-1}}, \qquad k=1,2,\ldots,n,\end{equation*}
with $k'=n+1-k.$ Here
\begin{align*}
& \tilde W_k=\tilde V_k\tilde U_k-\tilde V_{k-1}\tilde U_{k+1},\qquad \tilde Z_k=\tilde T_k\tilde U_k-\tilde T_{k+1}\tilde U_{k-1},
\end{align*}
with
\begin{align*}
&\tilde T_{2k}=(-1,0,\cdots,2k-2),  &&\tilde T_{2k-1}=(d_0,-1,0,\cdots,2k-3),\\
&\tilde U_{2k}=(0,1,\cdots,2k-1), &&\tilde U_{2k-1}=(d_0,0,1,\cdots,2k-2),\\
&\tilde V_{2k}=(1,2,\cdots,2k),&&\tilde V_{2k-1}=(d_0,1,2,\cdots,2k-1).
\end{align*}
The Pfaffian entries are defined as
\[
(d_0,i)=f_i(t),\qquad (i,j)=e^{-\alpha t}{\iint\limits_{-\infty<t_1<t_2<t}}e^{\alpha t_2}\left[ f_{i-1}(t_2)f_{j-1}(t_1)-f_{i-1}(t_1)f_{j-1}(t_2)\right] dt_1 dt_2,
\]
with
\begin{equation*}
f_i(t)=\sum_{j=1}^n(\zeta_j)^ib_j(t),\qquad\qquad 0<\zeta_1<\cdots<\zeta_n, \qquad b_j(t)=b_j(0)e^{\zeta_jt}>0.\]

\end{theorem}

 \section{Conclusion and discussions}\label{sec:con}
By applying Pfaffians to the analysis of the Novikov peakons, we verified that the Novikov peakon ODEs are no other than  Pfaffian identities. Motivated by the relation between the CH peakon ODEs and the Toda lattice, we gave a connection between the Novikov peakon ODEs and the finite B-Toda lattice. Furthermore, we proposed new generalizations 
 using the Pffafian technique  as our main tool. 

The DP equation is also an intriguing equation among the CH type equations.  Even though we succeeded in expressing its multipeakons in terms of Pfaffians, it seems to us  that the use of Pfaffians is not a natural tool to deal with the DP case. Recall that there are reciprocal links between the Novikov equation and a negative flow in the SK hierachy (BKP type), while the DP equation corresponds to a negative flow in the KK hierachy (CKP type). The study of the CKP type is a bit different from that for the BKP type. We address the corresponding result elsewhere \cite{chang2018degasperis}.
 
 Finally, we would like to remark that the research on the peakons and (bi)orthogonal polynomials is also an interesting topic. The CH peakons can be expressed by ordinary orthogonal polynomials \cite{beals2000multipeakons}. Symmetric orthogonal polynomials play a substantial role in the study of interlacing peakons of a two component modified CH equation \cite{chang2016multipeakons} . The DP peakons \cite{lundmark2005degasperis} are associated with Cauchy bimoment determinants, which induce the research on Cauchy biorthogonal polynomials and subsequent applications to random matrix models \cite{bertola2009cauchy,bertola2010cauchy}. The NV, GX peakons \cite{hone2009explicit,lundmark2014inverse} also correspond to Cauchy biorthogonal polynomials with different measures. In the present paper, the Pfaffian has been applied to the Novikov peakons. Actually, the skew-orthogonal polynomials first introduced in the theory of random matrices are expressed in terms of Pfaffians \cite{adler2000classical,adler1999pfaff,deift2000orthogonal,dyson1972class,mehta2004random}. There should be some underlying connections among these items, which deserve further study\footnote{We remind the readers that one piece of work \cite{chang2019partial} has been done while the current paper was under review. It turns out that the related random matrix model is the Bures ensemble. And a family of so-called partial-skew-orthogonal polynomials is introduced to act as a wave vector for an isospectral problem of the B-Toda lattice.}.
 
 \section{Acknowledgements}
We are grateful to Professor Jacek Szmigielski for his valuable comments on the generalized Novikov equation and the improvement of the presentation. We also thank a reviewer for posing some interesting questions for further study. This work was supported in part by the National Natural Science Foundation of China. X.C. was supported under the grant nos 11701550, 11731014. X.H. was supported under the grant nos 11331008, 11571358. J.Z. was supported under the grant no 11271362. 

 \begin{appendix}

\section{On the Pfaffian}\label{app_pf} 
 Determinants are well known, however, most people know little about Pfaffians \cite{Caieniello1973,hirota2004direct}.  This appendix contains some useful materials on Pfaffians. 

\textit{Pfaffian} was introduced by Cayley in 1852, who named it after Johann Friedrich Pfaff.  A Pfaffian is closely related to the determinant of an antisymmetric matrix. Let $A=(a_{i,j})$ be an $M\times M$ antisymmetric matrix leading to
\begin{equation*}
\det(A)=|(a_{i,j})|_{M\times M},\ \ a_{i,j}=-a_{j,i},\ 1\leq i,j\leq M.
\end{equation*}
 It is obvious that $\det(A)=0$ when $M$ is odd. And if $M$ is even, then $\det(A)$ is a square of a Pfaffian.
For example,
\begin{enumerate}[(1)]
\item For $M=2$, we have $\det(A)=(a_{1,2})^2=(Pf(1,2))^2$, and  $Pf(A)=Pf(1,2)=a_{1,2}$.
\item For $M=4$, we have $\det(A)=(a_{1,2}a_{3,4}-a_{1,3}a_{2,4}+a_{1,4}a_{2,3})^2=(Pf(1,2,3,4))^2$, and  $Pf(A)=Pf(1,2,3,4)=a_{1,2}a_{3,4}-a_{1,3}a_{2,4}+a_{1,4}a_{2,3}$.
\end{enumerate}

\subsection{Definition and properties}
We shall begin with the following definition in order to be self-contained even though there are some alternative definitions.
 \begin{define}\label{def:pfa}
Given the Pfaffian entries $Pf(i,j)$ satisfying $Pf(i,j)=-Pf(j,i)$, an N-order Pfaffian $ Pf(1,2,\cdots,2N)$ can be defined by 
\begin{equation}
Pf(1,2,\cdots,2N)=\sum_{P}(-1)^PPf(i_1,i_2)Pf(i_3,i_4)\cdots Pf(i_{2N-1},i_{2N}).
\end{equation}
 The summation means the sum over all possible combinations of pairs selected from $\{1,2,\cdots,2N\}$ satisfying
\begin{align*}
&i_{2l-1}<i_{2l+1},\qquad i_{2l-1}<i_{2l}.
\end{align*}
The factor $(-1)^P$ takes the value $+1 (-1)$ if the sequence $i_1, i_2, . . . , i_{2n}$ is an even (odd) permutation of $1, 2, . . . , 2n.$ Usually, there are the conventions that the Pfaffian of order $0$ is 1 and that for negative order is $0$.
\end{define}

Under the definition above, the relation between a Pfaffian and the determinant of the corresponding antisymmetric matrix was first rigorously proved by Muir \cite{muir1882treatise} in 1882.
 \begin{theorem}[Muir \cite{muir1882treatise}]
The square of an N-order Pfaffian $Pf(1,2,\cdots,2N)$ under Definition \ref{def:pfa}  is the determinant of the corresponding antisymmetric matrix:
\begin{equation*}
Pf(1,2,\cdots,2N)^2=|(a_{i,j})|_{N\times N},\ \ \ 1\leq i,j\leq 2N,
\end{equation*}
where $a_{i,j}=Pf(i,j).$
\end{theorem}

\begin{remark}\label{rem:pf1}
For a given antisymmetric matrix $A$ with the entries $a_{i,j}$, if we let $Pf(i,j)=a_{i,j}$, a Pfaffian will be obtained from the Definition \ref{def:pfa}.  The notation $Pf(A)$ is usually used to express the Pfaffian for the given antisymmetric matrix $A$. In this case, we have 
$$Pf(A)=Pf(1,2,\cdots,2N),$$
where  $Pf(i,j)=a_{i,j}$. 

So far, we have used two notations for Pfaffians. In fact, there are other notations. From now on, for our convenience, we will use a simple notation $(1,2,\cdots,2N)$ with the entries $(i,j)$ if there is no confusion caused. 
\end{remark}
\begin{remark}\label{rem:pf2}
On many occasions,  there appear not only numbers such as $1,2,..$, but also some characters $a,b,...$ (or $1^*,2^*...$) in the expressions of the Pfaffian. It will be clear if one get the corresponding antisymmetric matrix in his mind. For instance,  
\[(a_1,a_2,1,2,\cdots,2N-2)=Pf(B),\] 
where $B=(b_{i,j})$ is the antisymmetric matrix of order $2N\times 2N$ with the entries defined by
\[ b_{ij}=\left\{\begin{array}{ll}
(a_1,a_2),& i=1,\ j=2,\\
(a_1,j-2), & i=1,\ 3\leq j\leq 2N,\\
(a_2,j-2),& i=2,\ 3\leq j\leq 2N,\\
(i-2,j-2),& 3\leq i, j\leq 2N.
\end{array}\right.
\]
Another example is 
\[(1,2,\cdots,N,N^*,\cdots,2^*,1^*)=Pf(C),\] 
where
$C=(c_{i,j})$ is the antisymmetric matrix of order $2N\times 2N$ with the entries defined by
\[ c_{ij}=\left\{\begin{array}{ll}
(i,j),& 1\leq i,j\leq N,\\
(i,(2N+1-j)^*), & 1\leq i\leq N,\ N+1\leq j\leq 2N,\\
((2N+1-i)^*,(2N+1-j)^*),& \ N+1\leq i,j\leq 2N.
\end{array}\right.
\]
\end{remark}

A Pfaffian has some basic properties similar to those for a determinant.
\begin{prop}
\begin{enumerate}
\item Multiplication of a row and a column by a constant is equivalent to multiplication of the Pfaffian by the same constant.
That is, 
$$(1,2\cdots,c\cdot i,\cdots,2N)=c\cdot(1,2\cdots,i,\cdots,2N).$$
\item Simultaneous interchange of the two different rows and corresponding columns changes the sign of the Pfaffian.
That is, 
$$(1,2\cdots,i,\cdots,j,\cdots,2N)=-(1,2\cdots, j,\cdots,i,\cdots,2N).$$
\item A multiple of a row and corresponding column added to another row and corresponding column does not change the value of the Pfaffian.
That is, 
$$(1,2\cdots,i+c\cdot j,\cdots,j,\cdots,2N)=(1,2\cdots, i,\cdots,j,\cdots,2N).$$
\item 
A Pfaffian admits similar Laplace expansion, that is, for a fixed $i$, $1\leq i\leq 2N$,
\begin{eqnarray*}
(1,2,\cdots 2N)&=&\sum_{1\leq j\leq 2N,j\neq i}(-1)^{i+j-1}(i,j)(1,\cdots,\hat{i},\cdots,\hat{j},\cdots,2N),
\end{eqnarray*}
where $\hat{j}$ denotes that the index $j$ is omitted. 
If we choose $i$ as $1$ or $2N$, we have
\begin{eqnarray*}
(1,2,\cdots 2N)
&=&\sum_{j=2}^{2N}(-1)^j(1,j)(2,3,\cdots,\hat{j},\cdots,2N)\\
&=&\sum_{j=1}^{2N-1}(-1)^{j+1}(1,2,\cdots,\hat{j},\cdots,2N-1)(j,2N).
\end{eqnarray*}
Furthermore, if $(a_0,b_0)=0,$ we have the expansion
\begin{equation*}
(a_0,b_0,1,2,\cdots,2N)=\sum\limits_{1\leq j<k\leq 2N}(a_0,b_0,j,k)(1,2,\cdots,\hat j,\cdots,\hat k,\cdots,2N).
\end{equation*}

\end{enumerate} 
\end{prop}

\begin{remark}
\textit{Pfaffian} is a more general algebraic tool than \textit{determinant}. In fact, every N-order determinant can be expressed by a Pfaffian. Given a determinant
\begin{equation*}
B=|(b_{i,j})|_{N\times N},
\end{equation*}
then it can be expressed by a Pfaffian
 \begin{equation*}
B=(1,2,\cdots,N,N^*,\cdots,2^*,1^*),
\end{equation*}
where the Pfaffian entries $(i,j)$, $(i^*,j^*)$, $(i,j^*)$ are defined by
\begin{equation*}
(i,j)=(i^*,j^*)=0,\ \ (i,j^*)=b_{i,j},\ 1\leq i,j\leq N.
\end{equation*}
Conversely, it fails.
\end{remark}

\subsection{Formulae related to determinants and Pfaffians}
There are many interesting connections between some determinants and Pfaffians.

\begin{enumerate}[I.]
\item The result below is due to
de Bruijn \cite{de1955some}.
For the even case,
\begin{align}
\idotsint\limits_{-\infty<\tau_1<\cdots<\tau_{2N}<t} \det[\varphi_i(\tau_j)]_{i,j=1,\cdots,2N}d\tau_1\cdots d\tau_{2N}=(1,2,\cdots, 2N), \label{id_deBr1}
\end{align}
where
$$
(i,j)={\iint\limits_{-\infty<\sigma<\tau<t}}\left[ \varphi_i(\sigma)\varphi_j(\tau)-\varphi_j(\sigma)\varphi_i(\tau)\right] d\sigma d\tau.
$$
For the odd case,
\begin{align}
\idotsint\limits_{-\infty<\tau_1<\cdots<\tau_{2N+1}<t} \det[\varphi_i(\tau_j)]_{i,j=1,\cdots,2N+1}d\tau_1\cdots d\tau_{2N+1}=(d_0,1,2,\cdots, 2N+1), \label{id_deBr2}
\end{align}
where
$$
(d_0,i)=\int\limits_{-\infty<\sigma<t}\varphi_i(\sigma)d\sigma, \qquad (i,j)={\iint\limits_{-\infty<\sigma<\tau<t}}\left[ \varphi_i(\sigma)\varphi_j(\tau)-\varphi_j(\sigma)\varphi_i(\tau)\right] d\sigma d\tau.
$$

\item The Cauchy-Binet formula is a well-known result on minor-summation for a determinant. There also exist minor-summation formulae for a Pfaffian, which was given by Ishikawa and Wakayama \cite[Theorem 1]{ishikawa1995minor}. We restate it as follows:

For convenience, we will employ the notation $Pf(A)$ with a skew-symmetric matrix $A$ for a Pfaffian. For any matrix $B$ and some valid ascending index sets $I,J$, $B_I^J$ denotes the sub-matrix formed from the remaining rows $I$ and columns $J$ of $B$.

Let $k$ and $N$ be integers such that $k\leq N$. Let $R = (r_{il})_{1\leq k,1\leq l\leq N}$ be a $k\times N$ rectangular matrix.

If $k$ is even and $A = (a_{ij})_{1\leq i,j\leq N}$ is a skew-symmetric matrix of size $N$, then 
\begin{equation}
\sum_{I\in\left(\substack{[1,N]\\k}\right)}Pf(A_I^I)\det(R_I)=Pf(Q),\label{id:minor_even}
\end{equation}
where $Q=(q_{ij})_{1\leq i,j\leq k}=RAR^T$. More exactly, the entries of Q are given by
\[
q_{i,j}=\left\{\begin{array}{ll}
0,&i=j,\\
\sum\limits_{l_1=1}^n\sum\limits_{l_2=1}^nr_{il_1}a_{l_1l_2}r_{jl_2},&i\neq j.
\end{array}\right.
\]

If $k$ is odd and $\tilde A = (\tilde a_{ij})_{0\leq i,j\leq N}$ is a skew-symmetric matrix of size $N+1$, then 
\begin{equation}
\sum_{I\in\left(\substack{[1,N]\\k}\right)}Pf(\tilde A_{0,I}^{0,I})\det(R_I)=Pf(\tilde Q),\label{id:minor_odd}
\end{equation}
where  $\tilde Q=(\tilde q_{ij})_{0\leq i,j\leq k}$ is a skew-symmetric matrix of size $k+1$ with the entries
\[
\tilde q_{ij}=\left\{\begin{array}{ll}
0,&i=j,\\
\sum\limits_{l=1}^nr_{jl}\tilde a_{0l}&i=0,\  1\leq j\leq k,\\
\sum\limits_{l=1}^nr_{il}\tilde a_{l0} &j=0,\  1\leq i\leq k,\\
\sum\limits_{l_1=1}^n\sum\limits_{l_2=1}^nr_{il_1}\tilde a_{l_1l_2}r_{jl_2},&i\neq j, \ 1\leq i,j\leq n.

\end{array}\right.
\]
\begin{remark}
When $k=N$, \eqref{id:minor_even} reduces to the well-known formula
\[
Pf(A)\det(R)=Pf(RAR^T).
\]
\end{remark}

\begin{remark}
\eqref{id:minor_even} also implies the Cauchy-Binet formula for determinants.
Please see \cite[Corollary  2.2]{ishikawa1996applications} for more details.
\end{remark}


\end{enumerate}

\subsection{Formulae for special Pfaffians}
 The following formulae play important roles in our paper. 
\begin{enumerate}[I.]

\item
 A central role is played by Schur's Pfaffian identity\cite{ishikawa2006generalizations,schur1911uber} in enumerative combinatorics and representation theory, which reads
\begin{equation}
(1,2,\cdots, 2N)=\prod_{1\leq i<j\leq 2N}\frac{s_i-s_j}{s_i+s_j}, \qquad \text{where} \qquad (i,j)=\frac{s_i-s_j}{s_i+s_j},\label{id_schur}
\end{equation}
for even case. In the odd case, if $(a_0,i)=1$, we have 
\begin{equation}
(a_0,1,2,\cdots, 2N-1)=\prod_{1\leq i<j\leq 2N-1}\frac{s_i-s_j}{s_i+s_j}, \qquad \text{where} \qquad (i,j)=\frac{s_i-s_j}{s_i+s_j}.\label{id_schur_odd}
\end{equation}

\item
Some derivative formulae for some special Pfaffians \cite{hirota2004direct} hold.
\begin{enumerate}[(1)]
\item If the $x$-derivative of a Pfaffian entry $(i,j)$ satisfy $\frac{\partial}{\partial x}(i,j)=(i+1,j)+(i,j+1)$, then 
\begin{align}
\frac{\partial}{\partial t}(i_1,\cdots,i_{2N})&=\sum_{k=1}^{2N}(i_1,i_2,\cdots,i_k+1,\cdots,i_{2N}),\label{der1_gen}
\end{align}
which gives as a special case
\begin{align}
\frac{\partial}{\partial x}(1,2,\cdots,2N)&=(1,2,\cdots,2N-1,2N+1) .\label{der1}
\end{align}
If $\frac{\partial}{\partial x}(a_0,i)=(a_0,i+1),$ then
\begin{align}
\frac{\partial}{\partial x}(a_0,1,2,\cdots,2N-1)&=(a_0,1,2,\cdots,2N-2,2N) .\label{der1_odd}
\end{align}

Similarly, if $\frac{\partial}{\partial x}(i,j)=(i-1,j)+(i,j-1)$, then 
\begin{align}
\frac{\partial}{\partial t}(i_1,\cdots,i_{2N})&=\sum_{k=1}^{2N}(i_1,i_2,\cdots,i_k-1,\cdots,i_{2N}),\label{der1_2_gen}
\end{align}
which gives as a special case
\begin{align}
\frac{\partial}{\partial x}(1,2,\cdots,2N)&=(0,2,3,\cdots,2N) .\label{der1_2}
\end{align}
If $\frac{\partial}{\partial x}(a_0,i)=(a_0,i-1),$ then
\begin{align}
\frac{\partial}{\partial x}(a_0,1,2,\cdots,2N-1)&=(a_0,0,2,\cdots,2N-1) .\label{der2_odd}
\end{align}

\item If  $\frac{\partial}{\partial x} (i,j)=(a_0,b_0,i,j)$ and $(a_0,b_0)=0$, then
    \begin{align}
        \frac{\partial}{\partial x} (i_1,i_2,\cdots,i_{2N})=(a_0,b_0,i_1,i_2,\cdots,i_{2N}),\label{der2_1_gen}
    \end{align}
which gives as a special case
\begin{align}
    \frac{\partial}{\partial x} (1,2,\cdots,2N)=(a_0,b_0,1,2,\cdots,2N).\label{der2_1}
    \end{align}
    If $\frac{\partial}{\partial x} (i,j)=(a_0,b_0,i,j)$ , $\frac{\partial}{\partial x} (a_0,j)=(b_0,j)$, then
     \begin{align}
        \frac{\partial}{\partial x} (a_0,i_1,i_2,\cdots,i_{2N-1})=(b_0,i_1,i_2,\cdots,i_{2N-1}),\label{der2_2_gen}
    \end{align}
which gives as a special case
      \begin{equation}
    \frac{\partial}{\partial x} (a_0,1,2,\cdots,2N-1)=(b_0,1,2,\cdots,2N-1).\label{der2_2}
    \end{equation}
\end{enumerate}

\begin{remark}\label{rem:wg}
The first type formula can be named as the Wronski-type formula  since it is similar to the Wronskian determinant. The second type is referred as Gram-type formula. See \cite{hirota2004direct} for more details.
\end{remark}
\begin{remark}
Usually, these derivative formulae for the general cases are derived by induction. In the main body, we will use the induction to develop some extra derivative formulae.
\end{remark}
\end{enumerate}

\subsection{Bilinear identities on Pfaffians}
As is known, there exist many identities in the theory of determinants. As a more generalized mathematical tool than determinants, Pfaffians also allow various kinds of identities. We will give several common Pfaffian identities in the following, which can be found in \cite{hirota2004direct}.
\begin{align*}
&(a_1,a_2,\cdots,a_{2k},1,\cdots,2N)(1,2,\cdots,2N)\nonumber\\
=&\sum_{j=2}^{2k}(-1)^j(a_1,a_j,1,\cdots,2N)(a_2,a_3,\cdots,\hat{a_j},\cdots,a_{2k},1,\cdots,2N),\\
&(a_1,a_2,\cdots,a_{2k-1},1,\cdots,2N-1)(1,2,\cdots,2N)\nonumber\\
=&\sum_{j=1}^{2k-1}(-1)^{j-1}(a_j,1,\cdots,2N-1)(a_1,a_2,\cdots,\hat{a_j},\cdots,a_{2k-1},1,\cdots,2N).
\end{align*}
In the particular case of $k=2$, the above identities respectively become
\begin{align}
&(a_1,a_2,a_3,a_4,1,\cdots,2N)(1,2,\cdots,2N)\nonumber\\
=&\sum_{j=2}^{4}(-1)^j(a_1,a_j,1,\cdots,2N)(a_2,\hat{a}_j,a_{4},1,\cdots,2N),\label{pf1}\\
&(a_1,a_2,a_3,1,\cdots,2N-1)(1,2,\cdots,2N)\nonumber\\
=&\sum_{j=1}^{3}(-1)^{j-1}(a_j,1,\cdots,2N-1)(a_1,\hat{a}_j,a_{3},1,\cdots,2N),\label{pf2}
\end{align}
which are much more practical. 

Let us mention that Jacobi identity and Pl\"{u}cker relation \cite{aitken1959determinants,hirota2004direct} may be viewed as the particular cases of the above Pfaffian identities.


 \end{appendix}
    
\small
\bibliographystyle{abbrv}

\def\cydot{\leavevmode\raise.4ex\hbox{.}}
  \def\cydot{\leavevmode\raise.4ex\hbox{.}} \def\cprime{$'$}

\end{document}